\documentclass[12pt,a4paper]{article}
\usepackage{verbatim}
\usepackage{caption}
\usepackage{adjustbox}
\usepackage{amsfonts}
\usepackage{graphics}
\usepackage{amsmath}
\usepackage{times}
\usepackage{appendix}
\usepackage{color}
\usepackage{mathtools}
\usepackage{enumerate}
\usepackage{fancyhdr,latexsym,amsmath,amsfonts,amssymb,amsbsy,amsthm,url}
\usepackage{setspace}
\usepackage[margin=1in,footskip=0.25in]{geometry}
\usepackage{graphics,graphicx,epsfig}
\usepackage{breqn}
\usepackage{caption,subcaption, float,subfloat}
\usepackage{pdflscape}
\usepackage{hyperref}
\usepackage[ruled,vlined]{algorithm2e}
\newtheorem{theorem}{Theorem}[]
\newtheorem{lemma}{Lemma}[]
\newtheorem{corollary}{Corollary}[]

\newtheorem{definition}{Definition}[]
\newtheorem{proposition}{Proposition}[]
\newtheorem{assumption}{Assumption}[]

\usepackage[english]{babel}
\usepackage[dvipsnames]{xcolor}
\makeatletter
\providecommand{\keywords}[1]
{
  \small	
  \textit{Keywords:} #1
}

\renewcommand{\section}{
	\@startsection
	{section}
	{1}
	{0pt}
	{1.1\baselineskip}
	{0.2\baselineskip}
	{\sc \centering}
}

\renewcommand{\subsection}{
	\@startsection
	{subsection}
	{1}
	{0pt}
	{1.1\baselineskip}
	{0.2\baselineskip}
	{\sc \centering}
}

\renewcommand{\subsubsection}{
	\@startsection
	{subsubsection}
	{1}
	{0pt}
	{1.1\baselineskip}
	{0.2\baselineskip}
	{\sc \centering}
}

\makeatother
\usepackage{pgfplots}
\usepackage{pgfplotstable}
\pgfplotsset{width=8cm,compat=1.9}

\newcommand{\norm}[1]{\lVert #1 \rVert}
\newcommand{\bnorm}[1]{\bigg \lVert #1 \bigg \rVert}
\newcommand{\abs}[1]{\lvert #1 \rvert}
\newcommand{\babs}[1]{\bigg \lvert #1 \bigg \rvert}
\usepackage[flushleft]{threeparttable}
\usepackage{rotating,booktabs, multirow}
\usepackage{colortbl}
\usepackage{makecell,cellspace, caption}
\newcommand{\probability}[1]{\mathbb{P}{#1}}
\newcommand{\expectation}[1]{\mathbb{E}{#1}}
\SetKwComment{Comment}{/* }{ */}

\begin{document}
\title{\large\sc Multilevel Monte Carlo in Sample Average Approximation: Convergence, Complexity and Application}
\normalsize
\author{\sc{Devang Sinha} \thanks{Department of Mathematics, Indian Institute of Technology Guwahati, Guwahati-781039, Assam, India, e-mail: dsinha@iitg.ac.in}
\and \sc{Siddhartha P. Chakrabarty} \thanks{Department of Mathematics, Indian Institute of Technology Guwahati, Guwahati-781039, Assam, India, e-mail: pratim@iitg.ac.in,
Phone: +91-361-2582606, Fax: +91-361-2582649}}
\date{}
\maketitle

\begin{abstract}
In this paper, we examine the Sample Average Approximation (SAA) procedure within a framework where the Monte Carlo estimator of the expectation is biased. We also introduce Multilevel Monte Carlo (MLMC) in the SAA setup to enhance the computational efficiency of solving optimization problems. In this context, we conduct a thorough analysis, exploiting Cramér's large deviation theory, to establish uniform convergence, quantify the convergence rate, and determine the sample complexity for both standard Monte Carlo and MLMC paradigms. Additionally, we perform a root-mean-squared error analysis utilizing tools from empirical process theory to derive sample complexity without relying on the finite moment condition typically required for uniform convergence results. Finally, we validate our findings and demonstrate the advantages of the MLMC estimator through numerical examples, estimating Conditional Value-at-Risk (CVaR) in the Geometric Brownian Motion and nested expectation framework.
\end{abstract}
\section{Introduction}
\label{section:introduction}
The Sample Average Approximation (SAA) method is a numerical algorithm grounded in Monte Carlo principles, designed to address optimization problems, where the input data is uncertain. Specifically, it targets approximating the optimal solution for optimization problems, formulated as follows:

\begin{equation}
\label{saa:eq1}
\min_{x \in \mathcal{X}}\left \{ F(x) := \mathbb{E}[f(x,\zeta)] \right \} 
\end{equation}

Here, $\mathcal{X}\subseteq \mathbb{R}^d$ is assumed to be a finite-dimensional compact set, and $f(\cdot,\zeta)$ denotes the cost function, with $\zeta$ being a random vector. Typically, this optimization doesn't lend itself to analytical solutions, necessitating a Monte Carlo-based approach for approximation. The SAA method, known for its simplicity and robustness, has become a preferred tool among practitioners. SAA operates by solving an approximation of the original problem, defined as:
\begin{equation}
\label{saa:eq2}
    \min_{x \in \mathcal{X}}\left \{ F_{N}(x):= \frac{1}{N}\sum_{k = 1}^{N}f(x,\zeta_k)\right\}
\end{equation}
where, $\zeta_1,\dots,\zeta_N$ represent independent and identically distributed (i.i.d) samples of the random vector $\zeta$, drawn from its distribution. Extensive literature has documented the convergence of the approximate solution to the optimal one, with numerous references discussing the convergence and providing comprehensive surveys on SAA. In this domain the seminal work by in \cite{kleywegt2002sample} demonstrated that if $f(\cdot,\zeta)$ is Lipschitz continuous, SAA requires a computational complexity of $\mathcal{O}((d+\gamma)\epsilon^{-2}\log(\epsilon^{-1}))$ to achieve an $\epsilon$-optimal solution with probability $\epsilon^{\gamma}$, for some $\gamma >0$. Similar results were obtained for constrained stochastic programs and two-stage stochastic optimization problems in subsequent studies \cite{shapiro2021lectures}.\\
However, existing studies inherently assumed that the Monte Carlo estimator is unbiased \textit{i.e.}
\[
\expectation{\left[\frac{1}{N}\sum_{k = 1}^N f(x,\zeta_k)\right]} = F(x) ~~\forall ~x\in \mathcal{X}.
\]
This paper explores the impact on the computational complexity of the SAA procedure due to bias introduced in the estimator while sampling the random variable from its approximate distribution. 
The primary motivation for studying the SAA in the biased framework bears testimony to the fact that in most practical scenarios, the sampling of the random variable from its exact distribution is not always possible. For example, portfolio selection problems in finance \cite{hong2017kernel}, robust supervised learning in computer vision and speech recognition \cite{hu2020sample}, reinforcement learning in policy evaluation \cite{hu2020sample}, all belong to a class of conditional stochastic optimization problem, defined as,
\begin{equation}
    \label{saa:eq3}
    \min_{x \in \mathcal{X}}\bigg\{F(x) := \mathbb{E}_{\zeta}\bigg[f\bigg(\mathbb{E}_{\eta|\zeta}[g_{\eta}(x,\zeta)]\bigg)\bigg]\bigg \}.
\end{equation}
In the above setup, approximating the inner expectation using a Monte Carlo procedure introduces a bias in estimating the expectation. As the problem belonging to the class of stochastic composition optimization has been a long-standing challenge in science and engineering, extensive research studying SAA to solve them is available in the literature, see, e.g.\cite{wang2017stochastic,yermol1971general,ermoliev2013sample,dentcheva2017statistical}. However, for the most part, the primary focus of their discussion is towards studying the asymptotic properties of the estimator, deriving central limit formulae \cite{dentcheva2017statistical} and establishing the rate of convergence \cite{ermoliev2013sample}. The study performed in \cite{hu2020sample} shows that the computational complexity required for solving the conditional stochastic optimization problem is $\mathcal{O}((d+\gamma)\epsilon^{-4}\log(\epsilon^{-1}))$, if the cost function is Lipschtiz continuous and is $\mathcal{O}((d+\gamma)\epsilon^{-3}\log(\epsilon^{-1}))$, if the cost function is smooth, where the increase in the computational complexity can be attributed to the Monte Carlo approximation of the inner expectation, which in turn induces bias in the estimator. Another problem of practical importance falls in the realm of risk management. It has been well-documented in the literature that problems related to risk estimation can be composed as the optimization problem, see, e.g.\cite{rockafellar2000optimization,ruszczynski2006optimization}. For example, Conditional Value-at-Risk (CVaR), a highly sought-after coherent risk measure, can be formulated as,
\begin{equation}
    \label{saa:eq4}
    \text{CVaR}_{\theta}(\zeta) = \min_{x \in \mathbb{R}}\left\{F(x) = \expectation{\left[x + \frac{1}{1-\theta}(\zeta - x)_{+}\right]}\right\}.
\end{equation}
Taking cues from the studies performed by in \cite{gordy2010nested,broadie2011efficient}, the above problem can be seen as a conditional stochastic optimization problem. Moreover, from the perspective of finance, if we assume that the underlying asset driving the loss function is modelled on a stochastic differential equation, then using the numerical approximation technique to approximate the loss would induce bias in the Monte Carlo estimation, affecting the performance of the SAA procedure. To our knowledge, a generalised study on the biased approximation of the random variable in the SAA paradigm has not been explored. We intend to fill this gap in the paper. The significant contribution of the paper is presented below.
\begin{itemize}
    \item[\textbf{(1)}] In section \ref{section:monte_carlo_saa}, we conduct a thorough analysis of the SAA framework that operates by solving the approximation of the original problem \eqref{saa:eq1} given by,
    \begin{equation}
    \label{saa:eq5}
         \min_{x \in \mathcal{X}}\left \{ F_h^N(x):= \frac{1}{N}\sum_{k = 1}^{N}f(x,\zeta_h^k)\right\}
    \end{equation}
    where $\zeta_h^k$'s are i.i.d sample of $\zeta_h$ and $h$ is a bias parameter. Gaining insights from the study in \cite{hu2020sample}, we state and prove the uniform convergence result and analyse the sampling complexity associated with achieving $\epsilon-$optimal solution.  
    
    \item[\textbf{(2)}] We further extend the analysis to the MLMC paradigm in section \ref{section:multilevel_saa}. The supremacy of MLMC over standard Monte Carlo has well been established in the literature, see, e.g.\cite{giles2008multilevel,giles2008improved,giles2013antithetic,pages2018numerical}. Therefore, the extension to the multilevel framework appears natural. The underlying idea behind the MLMC-SAA is solving the approximation to the original optimization problem given by
    \begin{equation}
    \label{saa:eq6}
        \min_{x \in \mathcal{X}}\left\{F_L(x) = \sum_{\ell = 0}^L\frac{1}{N_{\ell}}\sum_{j = 1}^{N_{\ell}}\left(f(x,\zeta^j_{\ell})- f(x,\zeta^j_{\ell-1})\right)\right\},
    \end{equation}where, $\zeta_{\ell}$ is the level $\ell$ approximation of the random variable $\zeta$. We observe that under some mild regularity conditions, MLMC-SAA appears to have computational complexity equivalent to that achieved under the unbiased framework. 
    
    \item[\textbf{(3)}] In section \ref{section:rmse_analysis}, we undertake a root-mean-squared error (RMSE) analysis of the optimal values obtained by solving the SAA problem both in the standard Monte Carlo and MLMC context. To our knowledge, such analysis is not available in the existing literature. In this regard, we borrow tools from empirical process theory, which we briefly review in section \ref{section: preliminaries}. We also extended the result to estimate the root-mean-squared error of the optimal gap estimator, which deals with assessing the quality of a candidate solution. Consequently, we also derive the sample complexity necessary to obtain the optimal value with $\epsilon-$RMSE.
\end{itemize}
In section \ref{section: preliminaries}, we recall the assumptions, definitions and results required to analyse the remainder of the work. Finally, in section \ref{section:numerical_illustration}, we perform numerical illustrations studying the efficacy of our discussed procedure both in the standard Monte Carlo and MLMC paradigm.

\section{Assumptions and Preliminaries}
\label{section: preliminaries}
Let $(\Omega, \mathcal{F},\mathbb{P})$ be the complete probability space. In this space, let us consider the following stochastic optimization problem,
\begin{equation}
    \label{saa:eq6}
    \min_{x \in \mathcal{X}}\left \{ F(x) = \mathbb{E}[f(x,\zeta)] \right \}
\end{equation}
where $\mathcal{X} \subset \mathbb{R}^d$ is a finite-dimensional compact set, $\zeta$ is a random vector whose distribution is supported on the set $\Theta \in \mathbb{R}^s$, and $f \colon \mathcal{X} \times \Theta \rightarrow \mathbb{R}$ is the cost function. Throughout our discussion, we assume that  for all $x \in \mathcal{X}$, $f(\cdot,\zeta)$ is Borel-measurable in $\zeta$ and is also Lipschitz continuous, \textit{i.e.},
\begin{assumption}
\label{lip_assumption}
For all $x_1,x_2 \in \mathcal{X}$,
    \begin{equation}
    \label{eq:7}
   \abs{f(x_1,\zeta)- f(x_2,\zeta)}\leq L_f\norm{x_1-x_2}
\end{equation}

where $L_f$ is the Lipschitz constant for any given $\zeta$. 
\end{assumption}
Since we will be working in the bias framework, we let $\mathcal{B}$ be the set of bias parameters with $\mathcal{B}\cup\{0\}$ being the compact set and
\[
\forall m \in \mathbb{N}, \quad \frac{\mathcal{B}}{m} \subset \mathcal{B}.
\]
We let $\zeta_h$ be an approximation of the random variable $\zeta$ for some $h \in \mathcal{B}$ defined on the same probability space. Also, for a given $x \in \mathcal{X}$, we have the random variable $f(x,\zeta_h)$ and $f(x,\zeta)$ such that $\displaystyle{\expectation{[f(x,\zeta_h)]}\rightarrow \expectation{[f(x,\zeta)]}}$ as $h \rightarrow 0$. In order to strengthen the bias condition, for a $x \in \mathcal{X}$, we define the \textit{weak expansion error} as,
\begin{equation}
\label{eq:8}
    \expectation{[f(x,\zeta_h)]} = \expectation{[f(x,\zeta)]} + c_1h^{\alpha} +o(h^{\alpha})
\end{equation}
where $\alpha >0$ and we assume consistency of the expansion \textit{i.e.} $c_1 \neq 0$. Lastly, we assume the existence of a unique solution to the optimization problem \ref{saa:eq6}. Below, we present the definitions that would be relevant throughout our analysis and discussion.
\begin{definition}
    Let $\displaystyle{\mathfrak{p}^* := \min_{x \in \mathcal{X}}\left \{ F(x) = \mathbb{E}[f(x,\zeta)] \right \}}$, then $x_{\epsilon} \in \mathcal{X}$ is said to the $\epsilon$-optimal solution if 
    \[
    F(x_{\epsilon})\leq \mathfrak{p}^*+\epsilon
    \]
\end{definition}
\begin{definition}
    For $v \in (0,1)$ and $\norm{\cdot}$ defined on $\mathcal{X}$, $\{x_k\}_{k = 1}^{Q(v,\norm{\cdot},\mathcal{X})}$ is said to be a $v$-net of $\mathcal{X}$ if
    \begin{itemize}
        \item $x_k \in \mathcal{X}$ for all $k \in \{1,\dots,Q(v,\norm{\cdot},\mathcal{X})\}$.
        \item $\forall x \in \mathcal{X}$ there exists $k(x) \in \{ 1, \dots, Q(v,\norm{\cdot},\mathcal{X})\}$, such that $\lVert x - x_{k(x)} \rVert \leq v$.
    \end{itemize}
\end{definition}
In empirical process theory, $Q(v,\norm{\cdot},\mathcal{X})$ is considered as the covering number, with ball size $v$, on the space $\mathcal{X}$. The following result gives an upper bound of the covering number.
\begin{lemma}
   Let $\mathcal{X}$ has a finite diameter $\mathcal{D}_{\mathcal{X}}$, then for any $v \in (0,1)$, there exists a $v$-net of $\mathcal{X}$ and size of that $v$-net is bounded, \textit{i.e.}, $Q(v,\norm{\cdot},\mathcal{X}) \leq \mathcal{O}((\mathcal{D}_{\mathcal{X}}/v)^d)$. 
\end{lemma}

Moreover, the Cram\'{e}r's large deviation theorem will frequently used throughout our analysis and we state it here as a lemma based on the discussion in \cite{ruszczynski2003stochastic,hu2020sample}. 
\begin{lemma}
\label{cramer_ld}
    Let $X_1,\dots,X_N$ be i.i.d samples of zero-mean random variable $X$ with finite variance $\sigma^2$. For any $\epsilon>0$, it holds that 
    \begin{equation}
        \probability{\left(\frac{1}{N}\sum_{j = 1}^N X_j\geq \epsilon\right)}\leq \exp(-NI(\epsilon)),\nonumber
    \end{equation}
    where $I(\epsilon)$ is the rate function defined as $\displaystyle{I(\epsilon):= \sup_{t \in \mathbb{R}}\{t\epsilon - \log(M(t))\}}$, with $\displaystyle{M(t):=\expectation{[e^{tX}]}}$ being the moment generating function of $X$. Further, for any $\delta>0$, there exists $\epsilon_1>0$, such that for any $\epsilon \in (0,\epsilon_1)$, $\displaystyle{I(\epsilon)\geq \frac{\epsilon^2}{(2+\delta)\sigma^2}}$.
\end{lemma}
Further, since RMSE analysis relies heavily on results from the empirical process theory, we present below a brief overview and state the pertinent results without proof.
\subsection{Empirical Process theory}
In this section, we present some results from the empirical process theory that would be pertinent to RMSE analysis. To begin with, let us denote by.
\[
\mathfrak{F}:=\{f(x,\cdot) - \expectation{[f(x,\cdot)]}: x \in \mathcal{X}\} 
\]
the class of centred cost function indexed on $x \in \mathcal{X}$. Under the assumption that the function $f(x,\cdot)$ is Lipschitz with respect to $x$, it is fairly easy to observe that the centred cost function $f(x,\cdot) - \expectation{[f(x,\cdot)]}$ is also Lipschitz albeit with a larger Lipschitz constant. Further, let $\ell^{\infty}(\mathcal{X})$ be a metric space of all bounded functions from $\mathcal{X}$ to $\mathbb{R}$ endowed with the supremum norm, \textit{i.e.}, $\norm{f-g}_{\infty}:=\sup_{x\in \mathcal{X}}\abs{f(x)-g(x)}$ for all $f,g \in \ell^{\infty}(\mathcal{X})$. If we define, 
\[
\mathbb{F}_{N}(\cdot) = \sqrt{N}\left(\frac{1}{N}\sum_{k = 1}^{N}(f(\cdot,\bar{\zeta}^k) - \expectation{[f(\cdot,\bar{\zeta})]})\right)
\]
then $\mathbb{F}_{N}$ is an empirical process indexed on the decision set $\mathcal{X}$ and $\mathbb{F}_{N}\in\ell^{\infty}(\mathcal{X})$.
Since our study heavily relies on moment bounds, we present the results below that provide the necessary bounds. As the concept of covering numbers and bracketing numbers are relevant to these results, we present the definition of the bracketing numbers, where the covering number is already defined above.
\begin{definition}[Bracketing numbers \cite{wellner2013weak}]
    Given two functions $f_1$ and $f_2$, the bracket $[f_1,f_2]$ is the set of all functions $f$ with $f_1\leq f\leq f_2$. An $v-$bracket is a bracket $[f_1 ,f_2]$ such that $\norm{f_1-f_2}<v$. The bracketing number $Q_{[]}(v,\mathfrak{F},\norm{\cdot})$ is the minimum number of $v-$brackets needed to cover $\mathfrak{F}$.
\end{definition}
The next result gives an upper bound for the bracketing number.
\begin{lemma}[Theorem 2.7.1 \cite{wellner2013weak} and Lemma EC.7 \cite{lam2018bounding}]
    If the cost function $f(x,\bar{\zeta})$ is Lipschtiz with respect to $x$, and the decision space $\mathcal{X}\subseteq\mathbb{R}^d$ is compact, then for any $v>0$
    \[
    Q_{[]}(4v\norm{L^{\zeta}_{f}}_2,\mathfrak{F},\norm{\cdot})\leq Q(v,\mathcal{X},\norm{\cdot})
    \]
\end{lemma}
Further, because of Lemma 1, we have an upper bound of the bracketing numbers. The following two results provide the relevant moment bound necessary for our analysis.
\begin{lemma}[Lemma EC.9 \cite{lam2018bounding}]
Let $\Bar{f}(x,\bar{\zeta}): = \sup_{x\in\mathcal{X}}\vert f(x,\bar{\zeta}) - \expectation[f(x,\bar{\zeta})]\vert$. Then, for all $N$, we have
\[
\sqrt{N}\expectation {\left[ \sup_{x \in \mathcal{X}}\left\vert \frac{1}{N}\sum_{k = 1}^{N}f(x,\bar{\zeta}^k) - \expectation{[f(x,\bar{\zeta})]}\right\vert \right]} \leq  C\norm{\bar{f}(x,\bar{\zeta})}_2\int_{0}^1\sqrt{1 + \log(Q_{[ ]}(v\norm{\bar{f}(x,\bar{\zeta})}_2,\mathfrak{F},\norm{\cdot}_2))}
\]
\end{lemma}
\begin{lemma}[Lemma EC.10 \cite{lam2018bounding}]
For any $p\geq 2$ it holds that 
\begin{align*}
\sqrt{N}\left(\expectation {\left[ \sup_{x \in \mathcal{X}}\left\vert \frac{1}{N}\sum_{k = 1}^{N}f(x,\bar{\zeta}^k) - \expectation{[f(x,\bar{\zeta})]}\right\vert^p \right]}\right)^{\frac{1}{p}} \leq C\left(\sqrt{N}\expectation {\left[ \sup_{x \in \mathcal{X}}\left\vert \frac{1}{N}\sum_{k = 1}^{N}f(x,\bar{\zeta}^k) - \expectation{[f(x,\bar{\zeta})]}\right\vert \right]}\right)\\+ C\left(N^{\frac{1}{p}-\frac{1}{2}}\norm{\bar{f}(x,\bar{\zeta})}_p\right) 
\end{align*}
\end{lemma}
In both the lemmas above $C$ is a universal constant.

\section{Uniform Convergence and Sample Complexity}
In this section, we state and prove the uniform convergence of the biased Monte Carlo SAA and the MLMC-SAA. Further, we derive the sample complexity results to illustrate the preponderance of MLMC-SAA over Monte Carlo SAA and estimate the desired probability.
\subsection{Monte Carlo SAA}
\label{section:monte_carlo_saa}
This section starts by stating the assumptions necessary for our discussion.
\begin{assumption}
  \label{bias_assmpution}
  $\displaystyle{\sup_{x \in \mathcal{X}}\abs{\expectation[f(x,\zeta_h) - f(x,\zeta)]}\leq c_1h^{\alpha}}$ for some $h \in \mathcal{B}$ and $\alpha>0$.
\end{assumption}
\begin{assumption}\
\label{variance_assumption}
    $\displaystyle{\sigma^2 := \sup_{h\in \mathcal{B}}\mathbb{E}}[\sup_{x \in \mathcal{X}}|f(x,\zeta_h)|^2]< \infty$.
\end{assumption}
The above two assumptions illustrate the boundedness of the bias and variance, respectively. In the literature, these assumptions are commonly used for computational complexity analysis. Now, recall the optimization problem \eqref{saa:eq1}, \textit{i.e.}, 
\[
\min_{x \in \mathcal{X}}\left \{ F(x) = \mathbb{E}[f(x,\zeta)] \right \},
\]
the Monte Carlo SAA approximation of the above problem is given  as
\begin{equation}
\label{saa:eq9}
\min_{x \in \mathcal{X}}\left\{F_h^N(x):=\frac{1}{N}\sum_{j = 1}^{N}f(x,\zeta_h^j)\right\}.
\end{equation}
Suppose we let $x^*$ and $x_h^*$ be optimal solutions to \eqref{saa:eq1} and \eqref{saa:eq9} respectively, then we are interested in determining the probability of $x_h^*$ being an $\epsilon$-optimal solution to \eqref{saa:eq1}. More specifically, we intend to determine $\displaystyle{\probability{(F(x_h^*) - F(x^*)\leq \epsilon)}}$ for some $\epsilon>0$. In order to do so, we begin by establishing the uniform convergence property based on concentration inequality. For the proof of the next result, we follow the approach in \cite{shapiro2021lectures,hu2020sample}
\begin{theorem}(Uniform Convergence)
\label{theorem:uniform_convergence_mc_saa}
Suppose assumptions \ref{lip_assumption}, \ref{bias_assmpution} and \ref{variance_assumption} are satisfied, then for any $\delta>0$, there exists an $\epsilon_1>0$ such that for all $\epsilon \in (0,\epsilon_1)$,
\begin{equation}
\label{saa:eq10}
 \probability{\bigg(\sup_{x \in \mathcal{X}}\abs{F_h^N(x) - F(x)}>\epsilon\bigg)} \leq \mathcal{O}(1)\bigg(\frac{4L_f\mathcal{D}_{\mathcal{X}}}{\epsilon}\bigg)^d \exp{\bigg(\frac{-Nc_1^2h^{2\alpha}}{(\delta+2)\sigma^2}\bigg)}
\end{equation}
\end{theorem}
\begin{proof}
    We begin with the construction of a $v$-net in order to get rid of the supremum over $x$. For this consider a $v$-net on $\mathcal{X}$ such that, $\displaystyle{v = \frac{\epsilon}{4L_f}}$ and thus $\displaystyle{Q \leq \mathcal{O}(1)\bigg(\frac{4L_f\mathcal{D}_{\mathcal{X}}}{\epsilon}\bigg)^d}$. Further, by invoking the Lipschitz continuity of $f(\cdot,\zeta)$, we have
\[
\abs{f(x,\zeta_h) - f(x_k,\zeta_h)}\leq \frac{\epsilon}{4} \text{ and } \abs{F(x) - F(x_k)}\leq \frac{\epsilon}{4}
\]
Therefore, for any $x \in \mathcal{X}$ we have,
\begin{equation}
\abs{F_h^N(x) - F(x)} \leq \frac{\epsilon}{2} + \abs{F_h(x_k) - F(x_k)}
\leq \frac{\epsilon}{2} + \max_{k = 1,\dots,Q}\abs{F_h(x_k) - F(x_k)} \nonumber
\end{equation}
Consequently, we have,
\begin{equation}
\begin{split}
    \probability{\bigg(\sup_{x \in \mathcal{X}}\abs{F_h^N(x) - F(x)}>\epsilon\bigg)} &\leq \probability{\bigg(\max_{k = 1,\dots,Q}\abs{F_h^N(x_k) - F(x_k)}>\frac{\epsilon}{2}\bigg)}\\
    &\leq \sum_{k =1}^Q\probability{\bigg(\abs{F_h^N(x_k) - F(x_k)}>\frac{\epsilon}{2}\bigg)} \nonumber
\end{split}
\end{equation}
In order to examine the $\displaystyle{\probability{\bigg(\abs{F_h^N(x_k) - F(x_k)}>\frac{\epsilon}{2}\bigg)}}$, we define the random variable $Z_h^j(k)$ as $\displaystyle{Z_h^j(k) := Y_h^j(x_k) - F(x_k)}$, and let $\displaystyle{\expectation{[Z_h(k)]}}$ be its respective expectation. It is easy to observe that $\displaystyle{Z_h^j(k) - \expectation{[Z_h(k)]}}$ is zero-mean random variable. Now if $\displaystyle{0< \expectation{[Z_h(k)]}\leq c_1h^{\alpha} \leq \epsilon/4}$, then 
\begin{equation}
\begin{split}
    \probability{\bigg(F_h^N(x_k) - F(x_k)>\frac{\epsilon}{2}\bigg)} &=  \probability{\bigg(\frac{1}{N} \sum_{j = 0}^N Z_h^j(k)>\frac{\epsilon}{2}\bigg)}\\
    &\leq  \probability{\bigg(\frac{1}{N} \sum_{j = 0}^N (Z_h^j(k)-\expectation{[Z_h(k)])}>\frac{\epsilon}{4}\bigg)}\\
    &\leq \probability{\bigg(\frac{1}{N} \sum_{j = 0}^N (Z_h^j(k)-\expectation{[Z_h(k)])}>c_1h^{\alpha}\bigg)}.\nonumber
\end{split}    
\end{equation}
Now, as a consequence of lemma \ref{cramer_ld} and our assumption on the variance, we have, 
\begin{equation}
    \begin{split}
        \probability{\bigg(\frac{1}{N} \sum_{j = 0}^N (Z_h^j(k)-\expectation{[Z_h(k)])}>c_1h^{\alpha}\bigg)} &\leq \exp{\bigg(\frac{-Nc_1^2h^{2\alpha}}{(\delta+2)\mathbb{V}[Z_h(k)]}\bigg)}\\
        &\leq \exp{\bigg(\frac{-Nc_1^2h^{2\alpha}}{(\delta+2)\sigma^2}\bigg)} \nonumber
    \end{split}
\end{equation}
Similarly, if $\displaystyle{0<-\expectation{[Z_h(k)]}\leq c_1h^{\alpha} \leq \epsilon/4}$, then,
\[
\probability{\bigg(\frac{1}{N} \sum_{j = 0}^N (\expectation{[Z_h(k)]-Z_h^j(k))}>c_1h^{\alpha}\bigg)} \leq \exp{\bigg(\frac{-Nc_1^2h^{2\alpha}}{(\delta+2)\sigma^2}\bigg)} \nonumber
\]
Finally, putting everything together, we have,
\[
 \probability{\bigg(\sup_{x \in \mathcal{X}}\abs{F_h^N(x) - F(x)}>\epsilon\bigg)} \leq \mathcal{O}(1)\bigg(\frac{4L_f\mathcal{D}_{\mathcal{X}}}{\epsilon}\bigg)^d \exp{\bigg(\frac{-Nc_1^2h^{2\alpha}}{(\delta+2)\sigma^2}\bigg)}
\]

\end{proof}
The following corollary is the immediate consequence of the above result,
\begin{corollary}
\label{corollary_1}
    Let assumptions of Theorem \ref{theorem:uniform_convergence_mc_saa} hold, and additionally assume $c_1h^{\alpha}\leq \epsilon/8$, then for any $\delta>0$, there exists $\epsilon_1 >0 $ such that for all $\epsilon \in (0,\epsilon_1)$,
    \begin{equation}
        \probability{\bigg(F(x_h^*) - F(x^*)>\epsilon\bigg)} \leq \mathcal{O}(1)\bigg(\frac{8L_f\mathcal{D}_{\mathcal{X}}}{\epsilon}\bigg)^d \exp{\bigg(\frac{-N c_1^2h^{2\alpha}}{(\delta+2)\sigma^2}\bigg)}
    \end{equation}
\end{corollary}
\begin{proof}
    To begin with, observe that $F_h(x_h^*) - F_h(x^*)\leq 0$. Now
    \begin{equation}
    \label{saa:eq11}
        \begin{split}
            \probability{\bigg(F(x_h^*) - F(x^*)\geq\epsilon\bigg)}
            &=\probability{\bigg([F(x_h^*)-F_h^N(x_h^*)]+[F_h^N(x_h^*)-F_h(x^*)]+[F_h^N(x^*)-F(x^*)]\geq \epsilon\bigg)}\\
            &\leq \probability{\bigg(F(x_h^*)-F_h^N(x_h^*)\geq \epsilon/2\bigg)}+\probability{\bigg(F_h^N(x^*)-F(x^*)\geq \epsilon/2\bigg)}. \nonumber
        \end{split}
    \end{equation}
    Invoking Theorem \ref{theorem:uniform_convergence_mc_saa} with the condition $c_1 h^{\alpha}\leq \epsilon/8$, we get the desired result.
\end{proof}
Now, let us assume the cost of generating a single sample of $Y_h(x)$ is $\Bar{\eta}/h$ with $\Bar{\eta}$ begin some proportionality constant. As the simulation requires generating $N$ samples, the total computational cost will be $\mathcal{C}:=N\Bar{\eta}/h$. Also, let the probability of the solution to the Monte Carlo SAA problem being $\epsilon$-optimal to the original problem defined in \eqref{saa:eq1} be at least $(1 - \epsilon^{\gamma})$ for some $\gamma >0$. The next result illustrates the computational complexity required to achieve an $\epsilon$-optimal solution. 
\begin{theorem}
\label{theorem:computational_complexity_mc_saa}(Computational Complexity)
     Let $\displaystyle{\epsilon<1/e}$ and $\gamma>0$. Then the computational complexity for achieving the $\epsilon$-optimal solution with the probability at least $\displaystyle{1 - \epsilon^{\gamma}}$ is $\displaystyle{\mathcal{O}\left((d+\gamma)\log(\epsilon^{-1})\epsilon^{-\left(2+\frac{1}{\alpha}\right)}\right)}$.
\end{theorem}
\begin{proof}
    Since we require $\displaystyle{\probability{\bigg(F(x_h^*) - F(x^*)\leq\epsilon\bigg)}\leq 1 - \epsilon^{\gamma}}$, we need
    $\displaystyle{\probability{\bigg(F(x_h^*) - F(x^*)>\epsilon\bigg)}< \epsilon^{\gamma}}$. Observe that if we take, 
\[    
    N = \left \lceil \frac{\mathcal{O}(1)(\delta+2)\sigma^2}{c_1^2 h^{2\alpha}}\bigg[d \log\left(\frac{8L_f\mathcal{D}_{\mathcal{X}}}{\epsilon}\right) + \log\left(\frac{1}{\epsilon^{\gamma}}\right)\bigg]\right \rceil,
\]then 
    \[
    \mathcal{O}(1)\bigg(\frac{8L_f\mathcal{D}_{\mathcal{X}}}{\epsilon}\bigg)^d \exp{\bigg(\frac{-N c_1^2h^{2\alpha}}{(\delta+2)\sigma^2}\bigg)}\leq \epsilon^{\gamma}
    \]
    and as a consequence of corollary \ref{corollary_1}, we have the desired probability. Now, the computational cost is defined as $
    \displaystyle{\mathcal{C} = N\Bar{\eta}/h}$. Therefore, taking,
    \[
    N = \frac{\mathcal{O}(1)(\delta+2)\sigma^2}{c_1^2 h^{2\alpha}}\bigg[d \log\left(\frac{8L_f\mathcal{D}_{\mathcal{X}}}{\epsilon}\right) + \log\left(\frac{1}{\epsilon^{\gamma}}\right)\bigg] + 1,
    \] we get
    \[
    \mathcal{C} = \Bar{\eta}\left( \frac{\mathcal{O}(1)(\delta+2)\sigma^2}{c_1^2 h^{2\alpha+1}}\bigg[d \log\left(\frac{8L_f\mathcal{D}_{\mathcal{X}}}{\epsilon}\right) + \log\left(\frac{1}{\epsilon^{\gamma}}\right)\bigg] + \frac{1}{h}\right).
    \] If $h = c_2\epsilon^{1/\alpha}$, so that $c_1h^{\alpha}\leq\epsilon/8$ (required to achieve the bound in corollary\ref{corollary_1}), then for $\displaystyle{\epsilon < \frac{1}{e}}$, observe that
    \[
    \frac{1}{h}<c_3(\gamma+d)\log(\epsilon^{-1})\epsilon^{-\left(2 + \frac{1}{\alpha}\right)},
    \]
    where, $\displaystyle{c_3 = \frac{1}{c_2}}$. Further, observe that, for $\displaystyle{\epsilon < \frac{1}{e}}$,
    \[
    d \log\left(\frac{8L_f\mathcal{D}_{\mathcal{X}}}{\epsilon}\right) + \log\left(\frac{1}{\epsilon^{\gamma}}\right)< c_4 (d+\gamma)\log(\epsilon^{-1})
    \]
    where $\displaystyle{c_4 = (1 + \max\{0,\log(8L_f\mathcal{D}_{\mathcal{X}})\})}$. Finally, putting everything together, we get
    \[
    \mathcal{C}<\Bar{\eta}\left( \frac{\mathcal{O}(1)(\delta+2)\sigma^2}{c_1^2 h^{2\alpha+1}}c_4 (d+\gamma)\log(\epsilon^{-1}) + c_3(\gamma+d)\log(\epsilon^{-1})\epsilon^{-\left(2 + \frac{1}{\alpha}\right)} \right).
    \]Substituting $h = c_2 \epsilon^{1/\alpha}$ in the above inequality, we get
    \[
    \mathcal{C}<c_5\left((\gamma+d)\log(\epsilon^{-1})\epsilon^{-\left(2 + \frac{1}{\alpha}\right)}\right)
    \] with $\displaystyle{c_5 = \Bar{\eta}\left( \frac{\mathcal{O}(1)(\delta+2)\sigma^2}{c_1^2 c_2^{2\alpha+1}}c_4 + c_3\right)}$.
\end{proof}
The above analysis shows the effect bias parameter $h$ and the convergence rate $\alpha$ have on the computational complexity of the Monte Carlo SAA. It is easy to observe that as $\alpha \rightarrow \infty$, \textit{i.e.} as the bias converges to zero, the computational complexity tends to $\displaystyle{\mathcal{O}\left((\gamma + d)\log(\epsilon^{-1})\epsilon^{-2}\right)}$, same as the one observed with an unbiased estimator. The complexity results for the conditional stochastic optimization can also be recreated in the context of the above analysis. For instance, the mean-squared-error analysis performed by in section 3.2 of \cite{hu2020sample} shows that the order of bias convergence is $1/2$, \textit{i.e.} $\alpha = 1/2$. Substituting this $\alpha$ in our above analysis, we obtained the computational complexity similar to the one achieved by authors, \textit{i.e.}, $\displaystyle{\mathcal{O}\left((\gamma + d)\log(\epsilon^{-1})\epsilon^{-4}\right)}$.

\subsection{Multilevel Monte Carlo SAA}
\label{section:multilevel_saa}

The idea of Multilevel Monte Carlo (MLMC) originated from the work of Heinrich \cite{heinrich2001multilevel} and was furthered in the seminal work in \cite{giles2008multilevel}. The primary aim of the MLMC algorithm is to improve the computational complexity required by the standard Monte Carlo to achieve the desired root-mean-square error. The idea is to use the various levels of approximation of the random variable and construct a telescoping sum so that the approximation of the coarser level acts as a control variate for the finer level. Mathematically speaking, if we are to estimate $\displaystyle{\expectation{f(\zeta)}}$, where the level $\ell$ approximation of the random variable $\zeta$ is available for $\ell = 0,\dots, L$, then the MLMC approximation of the expectation is given as,
\begin{equation}
\label{saa:eq12}
     \expectation{[f(\zeta)]} \approx \sum_{\ell = 0}^L\frac{1}{N_{\ell}}\sum_{j =1}^{N_{\ell}}\left(f(\zeta_{\ell}) - f(\zeta_{\ell-1})\right), \text{ where } f(\zeta_{-1})\equiv 0.
\end{equation}
The analysis undertaken in \cite{giles2008multilevel} shows the effect of variance convergence, \textit{i.e.} $\mathbb{V}(f(\zeta_{\ell})-f(\zeta_{\ell-1}))$, on the computational complexity of the multilevel estimator. In particular, suppose $\beta$ denotes the order of variance convergence. The authors proved that the MLMC achieve the complexity of $\mathcal{O}(\epsilon^{-2})$ if $\beta>1$, $\displaystyle{\mathcal{O}(\epsilon^{-2}\left(\log(\epsilon)\right)^2)}$ if $\beta =1$, and $\displaystyle{\mathcal{O}(\epsilon^{-\left(2+\frac{1-\beta}{\alpha}\right)})}$ if $\beta <1$, where $\alpha$ is weak error rate.
This research was paramount as the subsequent investigations were directed toward constructing estimators with $\beta>1$. In this regard, various studies were performed to improve the computational complexity for the applications other than one discussed in \cite{giles2008multilevel}. For example, \cite{giles2019multilevel,haji2022adaptive} extended the nested expectation algorithm examined in \cite{broadie2011efficient} to the multilevel framework. From the optimization perspective, studies in \cite{dereich2021general,frikha2016multi,crepey2023multilevel} developed a multilevel extension of the stochastic approximation algorithm. All these studies showed how extending the multilevel framework can help with computational savings. Drawing motivation from these studies, we develop a multilevel extension of the Monte Carlo SAA to improve the computational cost associated with achieving the $\epsilon$-optimal solution. In this paper, we do not aim to construct an estimator with higher order variance convergence but rather study the effect of MLMC approximation of the expectation in the optimization problem \eqref{saa:eq1}. In the subsequent section, we define the MLMC-SAA optimization problem and extend the analysis from section \ref{section:monte_carlo_saa} to the MLMC framework.\\
\\
The basic idea behind the multilevel extension of the Monte Carlo SAA deals with the multilevel approximation of expectation associated with the minimisation problem. To this end, we assume the availability of level $\ell$ approximation of the random variable $\zeta$, denoted by $\zeta_{\ell}$, such that as $\ell \rightarrow \infty$, $\zeta_{\ell} \rightarrow \zeta$. Therefore, the multilevel extension of the optimization problem is given as
\begin{equation}
     \min_{x \in \mathcal{X}}\left\{F_L(x) = \sum_{\ell = 0}^L\frac{1}{N_{\ell}}\sum_{j = 1}^{N_{\ell}}\left(f(x,\zeta^j_{\ell})- f(x,\zeta^j_{\ell-1})\right)\right\}. \nonumber
\end{equation}
Let $x_L^*$ solve the optimization problem stated above and $x^*$ solve the original problem \textit{i.e.},\eqref{saa:eq1}. As before, we aim to determine the probability of $x_L^*$ being the $\epsilon$-optimal solution. We conduct an analysis similar to the one performed in section \ref{section:monte_carlo_saa} to study the uniform convergence and computational complexity in the MLMC paradigm. We begin our discussion by stating some technical assumptions,
\begin{assumption}
\label{assumption_mlmc_bias}
    $\displaystyle{E_{\ell} := \sup_{x \in \mathcal{X}} \abs{\mathbb{E}\left(f(x, \zeta_{\ell}) - f(x, \zeta)\right)}}\leq c_1h_{\ell}^{\alpha}$, 
\end{assumption}
\begin{assumption}
\label{assumption_mlmc_var}
     $\displaystyle{V_{\ell} := \expectation[\sup_{x \in \mathcal{X}}\left\vert f(x, \zeta_{\ell}) - f(x, \zeta_{\ell - 1 })\right\vert^2]}\leq c_2h_{\ell}^{\beta}$.
\end{assumption}
\begin{theorem}(Uniform Convergence)
\label{theorem:uniform_convergence_mlmc_saa}
Suppose \ref{lip_assumption},\ref{assumption_mlmc_bias} and \ref{assumption_mlmc_var} holds. Then for any $r>0$,
\begin{equation}
\label{saa:eq13}
    \probability{\bigg(\sup_{x \in \mathcal{X}}\abs{F_L(x) - F(x)}>\epsilon\bigg)} \leq \ \mathcal{O}(1)\bigg(\frac{4(2L+1)L_f\mathcal{D}_{\mathcal{X}}}{\epsilon}\bigg)^d \sum_{\ell = 0}^L \exp{\left(\frac{-N_{\ell}\epsilon^2 (m^{r}-1)^2}{16(m^{2r} m^{2 r\ell}(\delta+2)c_2 h_{\ell}^{\beta}}\right)}.
\end{equation}
\end{theorem}

\begin{proof}
We begin with the construction of a $v$-net in order to get rid of the supremum over $x$. We first pick a $v$-net on $\mathcal{X}$ such that, $\displaystyle{v = \frac{\epsilon}{4(2L+1)L_f}}$; thus $\displaystyle{Q \leq \mathcal{O}(1)\bigg(\frac{4(2L+1)L_f\mathcal{D}_{\mathcal{X}}}{\epsilon}\bigg)^d}$. Further, by invoking the Lipschitz continuity of $h(\cdot,\zeta)$, we have
\[
\abs{f_{L}(x) - f_{L}(x_k)}\leq \frac{\epsilon}{4} \text{ and } \abs{F(x) - F(x_k)}\leq \frac{\epsilon}{4}
\]
Therefore, for any $x \in \mathcal{X}$ we have,
\begin{equation}
\abs{F_L(x) - F(x)} \leq \frac{\epsilon}{2} + \abs{F_L(x_k) - F(x_k)}
\leq \frac{\epsilon}{2} + \max_{k = 1,\dots,Q}\abs{F_L(x_k) - F(x_k)} \nonumber
\end{equation}
Consequently, we have,
\begin{equation}
    \probability{\bigg(\sup_{x \in \mathcal{X}}\abs{F_L(x) - F(x)}>\epsilon\bigg)} \leq \probability{\bigg(\max_{k = 1,\dots,Q}\abs{F_L(x_k) - F(x_k)}>\frac{\epsilon}{2}\bigg)}\leq \sum_{k =1}^Q\probability{\bigg(\abs{F_L(x_k) - F(x_k)}>\frac{\epsilon}{2}\bigg)} \nonumber.
\end{equation}
Let us suppose $\displaystyle{\expectation{[F_L(x_k) - F(x_k)]}\leq \epsilon/4}$, then
\begin{equation}
   \probability{\bigg(F_L(x_k) -F(x_k)>\frac{\epsilon}{2}\bigg)} \leq \probability{\bigg(F_L(x_k) -F(x_k) - \expectation{[F_L(x_k) - F(x_k)] >\frac{\epsilon}{4}\bigg)}}\nonumber
\end{equation}
Now, under the multilevel paradigm, we define the random variable $Z^j_{\ell}(k)$ as follows,
\begin{equation}
    Z^j_{\ell}(k) = \begin{cases}
        f(x_k,\zeta^j_0) - F(x_k), \ell = 0\\
        f(x_k,\zeta^j_{\ell}) - f(x_k,\zeta^j_{\ell-1}), \ell = 1,\dots,L.
    \end{cases}\nonumber
\end{equation}
and denote $\expectation{[Z_{\ell}(k)]}$ as its expectation. Now, with the above notational considerations, we have
\begin{equation}
    \probability{\bigg(F_L(x_k) -F(x_k) - \expectation{[F_L(x_k) - F(x_k)] >\frac{\epsilon}{4}\bigg)}} = \probability{\bigg(\sum_{\ell = 0}^L \frac{1}{N_{\ell}}\sum_{j = 1}^{N_{\ell}} (Z_{\ell}^j(k) - \expectation{[Z_{\ell}(k)]})>\frac{\epsilon}{4}\bigg)}\nonumber.
\end{equation}
Now, consider the following sets,
\begin{equation}
    O = \left\{ \sum_{\ell = 0}^L \frac{1}{N_{\ell}}\sum_{j = 1}^{N_{\ell}} (Z_{\ell}^j(k) - \expectation{[Z_{\ell}(k)]})>\frac{\epsilon}{4}\right\} \text{ and } O_{\ell} = \left\{\frac{1}{N_{\ell}}\sum_{j = 1}^{N_{\ell}} (Z_{\ell}^j(k) - \expectation{[Z_{\ell}(k)]})>\frac{\epsilon }{4}\frac{(m^r -1)}{m^r}\frac{1}{m^{r\ell}}\right\} \nonumber
\end{equation}
and further observe that $\displaystyle{O \subseteq \bigcup_{\ell = 0}^L O_{\ell}}$. Then, by finite sub-additivity of the probability measure, we have,  $\displaystyle{\probability{(O)}\leq \sum_{\ell = 0}^L \probability{(O_{\ell})}}$. As $Z_{\ell} - \expectation{[Z_{\ell}(k)]}$ is a zero-mean random variable, and observing that $\mathbb{V}[Z_{\ell}(k)]\leq c_2h_{\ell}^{\beta}$, we have,
\begin{equation}
    \probability{(O_{\ell})}\leq \exp{\left(\frac{-N_{\ell}\epsilon^2 (m^{r}-1)^2}{16(m^{2r} m^{2 r\ell}(\delta+2)c_2 h_{\ell}^{\beta}}\right)} \nonumber
\end{equation}
Therefore, we have,
\begin{equation}
    \probability{\left(F_{L}(x_k) - F(x_k) >\frac{\epsilon}{2}\right)} \leq \sum_{\ell = 0}^L \exp{\left(\frac{-N_{\ell}\epsilon^2 (m^{r}-1)^2}{16m^{2r} m^{2 r\ell}(\delta+2)c_2 h_{\ell}^{\beta}}\right)}.\nonumber
\end{equation}
Similarly if $\displaystyle{\expectation{[F_L(x_k) - F(x_k)]}\geq -\epsilon/4}$, then
\begin{equation}
     \probability{\left(F(x_k)-F_{L}(x_k) >\frac{\epsilon}{2}\right)} \leq \sum_{\ell = 0}^L \exp{\left(\frac{-N_{\ell}\epsilon^2 (m^{r}-1)^2}{16m^{2r} m^{2 r\ell}(\delta+2)c_2 h_{\ell}^{\beta}}\right)},\nonumber
\end{equation}
and hence, we have,
\begin{equation}
    \probability{\left(|F_{L}(x_k) - F(x_k)| >\frac{\epsilon}{2}\right)} \leq \sum_{\ell = 0}^L \exp{\left(\frac{-N_{\ell}\epsilon^2 (m^{r}-1)^2}{16m^{2r} m^{2 r\ell}(\delta+2)c_2 h_{\ell}^{\beta}}\right)}.\nonumber
\end{equation}
Putting everything together, we have,
\begin{equation}
    \probability{\bigg(\sup_{x \in \mathcal{X}}\abs{F_L(x) - F(x)}>\epsilon\bigg)} \leq \ \mathcal{O}(1)\bigg(\frac{4(2L+1)L_f\mathcal{D}_{\mathcal{X}}}{\epsilon}\bigg)^d \sum_{\ell = 0}^L \exp{\left(\frac{-N_{\ell}\epsilon^2 (m^{r}-1)^2}{16m^{2r} m^{2 r\ell}(\delta+2)c_2 h_{\ell}^{\beta}}\right)}.\nonumber
\end{equation}
\end{proof}
The following corollary is the immediate consequence of the above result.
\begin{corollary}
\label{corollary_2}
     Let assumptions of Theorem \ref{theorem:uniform_convergence_mlmc_saa} holds, and additionally assume $\displaystyle{\sup_{x \in \mathcal{X}}\abs{\expectation{[F_L(x) - F(x)]}}}\leq \epsilon/8$, then for any $\delta>0$, there exists $\epsilon_1 >0 $ such that for all $\epsilon \in (0,\epsilon_1)$ and for some $r > 0$,
    \begin{equation}
    \label{saa:eq14}
        \probability{\bigg(F(x_L^*) - F(x^*)>\epsilon\bigg)} \leq \mathcal{O}(1)\bigg(\frac{8(2L+1)L_f\mathcal{D}_{\mathcal{X}}}{\epsilon}\bigg)^d \sum_{\ell = 0}^L \exp{\left(\frac{-N_{\ell}\epsilon^2 (m^{r}-1)^2}{64m^{2r} m^{2 r\ell}(\delta+2)c_2 h_{\ell}^{\beta}}\right)}.
    \end{equation}
\end{corollary}
The proof of the above corollary follows the line of argument similar to the proof of corollary \ref{corollary_1} and is therefore skipped.\\

Now that we have established the convergence results, we pivot our focus towards the central motif of this article, namely, computational cost (sampling complexity). As previously indicated, let $\Bar{\eta}/h_{\ell}$ denote the cost associated with generating a single sample of $f(x,\zeta_{\ell})$ at level $\ell$. Thus, the computational expenditure for generating $N_{\ell}$ samples would amount to $\Bar{\eta}N_{\ell}/h_{\ell}$. Given that we generate samples across levels $\ell = 0,\dots,L$, the aggregate computational cost is delineated as:
\begin{equation}
\label{saa:eq15}
    \mathcal{C}_{mlmc}^{saa}:=\Bar{\eta}\sum_{\ell = 0}^L \frac{N_{\ell}}{h_{\ell}}.
\end{equation}
The following result estimates the computational complexity associated with the multilevel estimator.
\begin{theorem}(Computational Complexity)
    \label{theorem:compuational_complexity_mlmc_saa}
    Let the probability of the solution to the multilevel SAA problem being $\epsilon$-optimal to the original problem be at least $1-\epsilon^{\gamma}$, for some $\gamma >0$. Further, let the assumptions of Theorem \ref{theorem:uniform_convergence_mlmc_saa} and corollary \ref{corollary_2} hold. Also, assume the existence of positive constants $\alpha, \beta$,$r$ and $m$ with $\alpha\geq 1/2$ and $m>2$. Then, the computational complexity associated with achieving the $\epsilon$-optimal solution is
    \begin{equation}
    \label{saa:eq16}
        \mathcal{C}_{mlmc}^{saa} = \begin{cases}
            \mathcal{O}\left((\gamma+d)\epsilon^{-2}\log(\epsilon^{-1})\right),&\text{ for } \beta >1 .\\
            \mathcal{O}\left((\gamma+d)\epsilon^{-\left(2+\frac{\log(2)}{\alpha\log(m)}\right)}\log(\epsilon^{-1})\right), &\text{ for } \beta = 1.\\
           \mathcal{O}\left((\gamma+d)\epsilon^{-\left(2+\frac{1-\beta}{\alpha}\right)}\log(\epsilon^{-1})\right),& \text{ for } \beta < 1.\\
            
        \end{cases}
    \end{equation}
\end{theorem}
\begin{proof}
    Following the same line of argument as in Theorem \ref{theorem:computational_complexity_mc_saa}, observe that if 
\begin{equation}
    N_{\ell} = \left \lceil \frac{64 m^{2r} m^{2 r\ell}(\delta+2)c_2 h_{\ell}^{\beta}}{\epsilon^2 (m^{r}-1)^2}\log\left(\frac{\mathcal{A}}{\epsilon^{\gamma}}(L+1)\right)\right \rceil, \text{ where } \mathcal{A} = \mathcal{O}(1)\bigg(\frac{8(2L+1)L_f\mathcal{D}_{\mathcal{X}}}{\epsilon}\bigg)^d,\nonumber
\end{equation}
then $ \probability{\bigg(F(x_L^*) - F(x^*)>\epsilon\bigg)} \leq \epsilon^{\gamma}$. Therefore, we now have a formulation for the number of samples required on various levels of resolution. As for the number of levels and computational complexity, we separately analyse different cases. We present below a general expression for the computational cost based on the formulation in \eqref{saa:eq15} that would be relevant throughout our analysis.
\begin{equation}
\label{saa:eq17}
\mathcal{C}_{mlmc}^{saa}\leq\Bar{\eta}\left(\frac{64 m^{2r} (\delta+2)c_2 }{\epsilon^2 (m^{r}-1)^2}\log\left(\frac{\mathcal{A}}{\epsilon^{\gamma}}(L+1)\right)\sum_{\ell = 0}^Lm^{2 r\ell}h_{\ell}^{\beta-1} + \sum_{\ell = 0}^L \frac{1}{h_{\ell}}\right)
\end{equation}
To begin with, let,
\[
L = \left \lceil \frac{\log(8c_1 h_0 ^{\alpha}\epsilon^{-1})}{\alpha\log(m)}\right \rceil,
\]
then, 
\begin{align*}
L+1 &\leq \frac{\log(8c_1h_0^{\alpha}\epsilon^{-1})}{\alpha\log(m)} + 2 
\leq c_3(\log(\epsilon^{-1}))
\end{align*}
where $\displaystyle{c_3 = \left(\frac{1}{\alpha\log(m)}+\max \left\{ 0,\frac{\log(8c_1h_0^{\alpha})}{\alpha\log(m)}\right\}+2\right)}$. Also, for the above $L$, it is easy to observe that,
\[
\frac{\epsilon}{8c_1m^{\alpha}}\leq h_L^{\alpha} < \frac{\epsilon}{8c_1},
\]
thereby satisfying the bias conditions required for uniform convergence. Further, 
\[
\log\left(\frac{\mathcal{A}}{\epsilon^{\gamma}}(L+1)\right) \leq c_4(\gamma+d)\log(\epsilon^{-1}) 
\]
where $c_4 = (3+\log(c_316L_f\mathcal{D}_{\mathcal{X}}))$. And also, for $\alpha\geq 1/2$ and $\epsilon < 1/e$, 
\[
    \sum_{\ell = 0}^L \frac{1}{h_{\ell}}< \epsilon^{-2}c_5 < (\gamma + d)\epsilon^{-2}\log(\epsilon^{-1})c_5 
\]
with $\displaystyle{c_5 = \frac{(8c_1)^{1/\alpha}m^2}{m-1}}$. Now all that is left for us to analyse is $\displaystyle{\sum_{\ell = 0}^Lm^{2 r\ell}h_{\ell}^{\beta-1}}$. We will handle this case by case.\\
\textbf{Case 1 ($\beta = 1$):} For $\beta  = 1$, we have $\displaystyle{\sum_{\ell = 0}^Lm^{2 r\ell}h_{\ell}^{\beta-1} = \frac{m^{2r(L+1)}-1}{m^{2r}-1} \leq \frac{m^{2r(L+1)}}{m^{2r}-1}}$. Substituting the upper bound of $L+1$, we get 
\[
\frac{m^{2r(L+1)}}{m^{2r}-1} \leq (8c_1h_0^{\alpha})^{2r/\alpha} \epsilon^{-2r/\alpha}\frac{m^{2r}}{m^{2r}-1}.
\]
Choosing $\displaystyle{r = \frac{\log(2)}{2\log(m)}}$ we get that $\displaystyle{\frac{m^{2r}}{m^{2r}-1} = 2}$. Consequently, we have,
\[
\frac{m^{2r(L+1)}}{m^{2r}-1} \leq c_6\epsilon^{-\frac{\log(2)}{\log(m)\alpha}} 
\]
where $\displaystyle{c_6 = 2(8c_1h_0^{\alpha})^{\frac{\log(2)}{\log(m)\alpha}} }$.\\
\textbf{Case 2 ($\beta > 1$):} For this case we assume $2r<\beta-1$. Then we have
\begin{align*}
    \sum_{\ell = 0}^Lm^{2r\ell}h_{\ell}^{\beta - 1} &= h_0^{\beta-1}\sum_{\ell = 0}^L \frac{1}{m^{(\beta-1-2r)\ell}}\\
        &\leq h_0^{\beta-1}(1 - m^{-(\beta-1-2r)})^{-1}.
\end{align*}
As before, choosing $\displaystyle{r = \frac{\beta-1}{2} - \frac{\log(2)}{\log(m)}}$, we get $\displaystyle{(1 - m^{-(\beta-1-2r)})^{-1} = 2}$. Consequently, we have  
\[\displaystyle{\sum_{\ell = 0}^Lm^{2r\ell}h_{\ell}^{\beta - 1}\leq2h_0^{\beta-1}}.\]
\textbf{Case 3 ($\beta<1$):} In this case choosing $\displaystyle{r = \frac{1-\beta}{2} - \frac{\log(2)}{\log(m)}}$, we have,
\begin{align*}
        \sum_{\ell = 0}^Lm^{2r\ell}h_{\ell}^{\beta - 1} &= h_L^{-(1-\beta)}\sum_{\ell = 0}^L \frac{1}{m^{(1-\beta-2r)\ell}}\\
        &< h_L^{-(1-\beta)}(1 - m^{-(1-\beta-2r)})^{-1}\\
        &\leq 2\epsilon^{-(1-\beta)/\alpha}(8c_1)^{(1-\beta)/\alpha}
\end{align*}
By collating everything together, we get the desired result.
\end{proof}
The above analysis illustrates how the Multilevel approximation of expectations can yield enhancements in computational complexity. We observe how the MLMC-SAA performs better than the Monte Carlo SAA, achieving the unbiased level of performance for $\beta>1$. As in the case of the standard MLMC analysis, we see the effect of variance convergence on the computational cost, \textit{i.e.}, introducing the estimator with high order variance convergence can affect the overall performance of the MLMC-SAA. For instance, the variance analysis carried out in  \cite{hu2020sample} shows that $\beta = 1/2$ and as $\alpha = 1/2$, consequently we observe the computational complexity of $\displaystyle{\mathcal{O}\left((\gamma+d)\epsilon^{-3}(\log(\epsilon^{-1}))\right)}$, similar to the one achieved by a smooth function.\\
\\
Despite the theoretical analysis elucidating the uniform convergence of both standard and multilevel Monte Carlo in the biased framework, the practical realisation of such integration poses notable challenges. One challenge stems from relatively strong but unavoidable assumptions of the existence of a finite valued moment-generating function in the neighbourhood of zero. This assumption is not necessary to prove the convergence but is essential in order to determine the sample complexity required to achieve an exponential rate of convergence. Although in \cite{banholzer2022rates} studied the convergence in the almost sure sense, the extension to the biased setup is beyond the scope of this paper. For now, we direct our attention to RMSE analysis, consequently deriving sample complexity without such strong assumptions.

\section{Root Mean Squared Error Analysis}
\label{section:rmse_analysis}
In this section, we undertake the RMSE analysis of the biased SAA problem both in the context of standard as well as multilevel Monte Carlo setting. We start our discussion with a technical lemma that would be essential in the remainder of the study.
\begin{lemma}
    Suppose $0<\gamma<\infty$ and $c,d >0 \in \mathbb{R}$. Then for all $x \in (0,d]$, there exist a constant $\mathfrak{K}$ such that
    \[
    x\left(1 + \sqrt{c\log\left(\frac{d}{x}\right)}\right)\leq \mathfrak{K} x^{\frac{1}{1+\gamma}}
    \]
\end{lemma}
\begin{proof}
Let $\displaystyle{\Bar{\gamma} = 1 -\frac{1}{1+\gamma}}$ then $0<\Bar{\gamma}<1$. Consider the function,
\[
g(x) = \begin{cases}
    x^{\Bar{\gamma}}\left(1 + \sqrt{c\log\left(\frac{d}{x}\right)}\right)&,x \in (0,d]\\
    0 &,x = 0
\end{cases}
\]
then it is easy to see that $g(x)\geq 0$ is continuous on $[0,d]$. Since $[0,d]$ is compact, we have $x^* \in [0,d]$ such that $\displaystyle{g(x^*)  = \max_{x \in [0,d]}g(x)}$. The results follows by letting $\mathfrak{K} = g(x^*)$.

\end{proof}

\subsection{Standard Monte Carlo Paradigm}
The next result provides an RMSE bound for the optimal value of the SAA problem defined in equation \eqref{saa:eq5}, defined as,
\begin{equation}
    \text{RMSE}:=\bnorm{\min_{x \in \mathcal{X}}\frac{1}{N}\sum_{k = 1}^N f(x,\zeta_h^k) -\mathfrak{p}^*}_2 
\end{equation}
where $\mathfrak{p}^*$ is the optimal value of the original SAA problem.
\begin{theorem}
    Suppose assumptions \ref{lip_assumption}, \ref{bias_assmpution} and \ref{variance_assumption} holds. Then for any $N\in \mathbb{N}$.
\[
 \bnorm{\min_{x \in \mathcal{X}}\frac{1}{N}\sum_{k = 1}^N f(x,\zeta_h^k) -\mathfrak{p}^*}_2 \leq c_1h^{\alpha} + \mathfrak{c}_3\frac{\sigma}{\sqrt{N}}.
\]
\end{theorem}
\begin{proof}
Let $\displaystyle{\mathfrak{p}^{*,h} =\min_{x \in \mathcal{X}}\expectation[f(x,\zeta_h)] }$. Observe that,
\begin{align*}
    \bnorm{\min_{x \in \mathcal{X}}\frac{1}{N}\sum_{k = 1}^N f(x,\zeta_h^k) -\mathfrak{p}^*}_2 &= \bnorm{\min_{x \in \mathcal{X}}\frac{1}{N}\sum_{k = 1}^N f(x,\zeta_h^k) -\mathfrak{p}^{*,h}+ \mathfrak{p}^{*,h} - \mathfrak{p}^*}_2\\
    &\leq  \bnorm{\min_{x \in \mathcal{X}}\frac{1}{N}\sum_{k = 1}^N f(x,\zeta_h^k) -\mathfrak{p}^{*,h}}+ \bnorm{\mathfrak{p}^{*,h} - \mathfrak{p}^*}_2
\end{align*}
For the second term in the above inequality, we have,
\[\displaystyle{\bnorm{\mathfrak{p}^{*,h} - \mathfrak{p}^{*}}_2\leq \left(\expectation{\left[\sup_{x\in\mathcal{X}}|\expectation{[f(x,\zeta_h)]}-\expectation{[f(x,\zeta)]}|^2\right]}\right)^{1/2}}\leq c_1h^{\alpha}.\]
As for the first term, we have,
\begin{align*}
    \bnorm{\min_{x \in \mathcal{X}}\frac{1}{N}\sum_{k = 1}^N f(x,\zeta_h^k) -\min_{x \in \mathcal{X}}\expectation[f(x,\zeta_h)]}_2 &= \left(\expectation{\left[\left(\min_{x \in \mathcal{X}}\frac{1}{N}\sum_{k = 1}^N f(x,\zeta_h^k) -\min_{x \in \mathcal{X}}\expectation[f(x,\zeta_h)]\right)^2\right]}\right)^{1/2}\nonumber\\
    &\leq \left(\expectation{\left[\sup_{x \in \mathcal{X}}\left(\babs{\frac{1}{N}\sum_{k = 1}^N f(x,\zeta_h^k) -\expectation[f(x,\zeta_h)]}\right)^2\right]}\right)^{1/2}\nonumber
\end{align*}
Now for a given $h \in \mathcal{B}$, we define $\mathbb{F}_{N}^h$ as,
\[
\mathbb{F}^h_{N}(\cdot) = \sqrt{N}\left(\frac{1}{N}\sum_{k = 1}^{N}(f(\cdot,\zeta_h^k) - \expectation{[f(\cdot,\zeta_h)]})\right),
\] 
is an empirical process. Then, under the assumptions \ref{bias_assmpution} and \ref{variance_assumption} and a direct application of Lemma 3,4 and 5, we have
\[
\sqrt{\expectation{\left[\sup_{x \in \mathcal{X}}\left(\babs{\frac{1}{N}\sum_{k = 1}^N f(x,\zeta_h^k) -\expectation[f(x,\zeta_h)]}\right)^2\right]}} \leq \mathfrak{c}_3\frac{\sigma}{\sqrt{N}}.
\]
where $\displaystyle{\mathfrak{c}_3 = C\int_{0}^1\sqrt{1 + \log(Q_{[ ]}(v\norm{\bar{f}(x,\bar{\zeta})}_2,\mathfrak{F},\norm{\cdot}_2))}<\infty}$. By collating everything and reassigning constants, we get the desired result.
\end{proof}
\begin{corollary}
    The computational complexity required for $\displaystyle{\bnorm{\min_{x \in \mathcal{X}}\frac{1}{N}\sum_{k = 1}^N f(x,\zeta_h^k) -\mathfrak{p}^*}_2\leq \epsilon}$ is $\mathcal{O}\left(\epsilon^{-\left(2+\frac{1}{\alpha}\right)}\right)$.
\end{corollary}
\begin{proof}
    Take $\displaystyle{h = \mathcal{O}(\epsilon^{-\frac{1}{\alpha}}})$ and $\displaystyle{N = \mathcal{O}(\epsilon^{-2})}$. Since the computational cost is given as $\displaystyle{\frac{N}{h}}$, we get the desired result.
\end{proof}
\subsection{Multilevel Monte Carlo Paradigm}
In the multilevel setting, we let, let $g(x,\bar{\zeta}_{\ell}):= f(x,\zeta_{\ell}) - f(x,\zeta_{\ell-1})$, and define $\mathfrak{p}^{*,L}: = \min_{x\in \mathcal{X}}\expectation{[f(x,\zeta^L)]} = \min_{x\in \mathcal{X}}\sum_{\ell =0}^L \expectation{[g(x,\bar{\zeta}_{\ell})]}$. Further, we define, 
\[\hat{\mathfrak{p}}^{*,L} : = \min_{x \in \mathcal{X}}\sum_{\ell = 0}^L \frac{1}{N_{\ell}}\sum_{k = 1}^{N_{\ell}}g(x,\bar{\zeta}^k_{\ell})\] 
as the Monte Carlo approximation of $\mathfrak{p}^{*,L}$. The RMSE error in this case is given as,
\begin{equation}
    \text{RMSE}_{MLMC}:= \bnorm{\min_{x \in \mathcal{X}}\sum_{\ell = 0}^L \frac{1}{N_{\ell}}\sum_{k = 1}^{N_{\ell}}g(x,\bar{\zeta}^k_{\ell}) - \mathfrak{p}^*}_2 = \bnorm{ \hat{\mathfrak{p}}^{*,L}-\mathfrak{p}^{*}}_2.
\end{equation}
The following result provides an RMSE error bound for the optimal value obtained by solving MLMC-SAA defined in \eqref{saa:eq6}.
\begin{theorem}
    Suppose assumptions \ref{lip_assumption}, \ref{assumption_mlmc_bias} and \ref{assumption_mlmc_var} holds. Then for any $0<a<\infty$, $L\geq2$.
    \[
    \bnorm{ \hat{\mathfrak{p}}^{*,L}-\mathfrak{p}^{*}}_2 \leq  c_1h_{L}^{\alpha} +2c_2\bar{\mathfrak{c}} \sum_{\ell = 0}^L\frac{ h_{\ell}^{\bar{\beta}/2}}{\sqrt{N_{\ell}}}
    \]
    where $\bar{\beta} = \beta\frac{1}{1+a}$.
\end{theorem}
\begin{proof}
To begin with, observe that from triangle inequality, we have,
\begin{align*}
\bnorm{ \hat{\mathfrak{p}}^{*,L}-\mathfrak{p}^{*}}_2 \leq \bnorm{\hat{\mathfrak{p}}^{*,L} - \mathfrak{p}^{*,L}}_2 + \bnorm{\mathfrak{p}^{*,L} - \mathfrak{p}^*}_2. 
\end{align*}
As before, we have,
\[\displaystyle{\bnorm{\mathfrak{p}^{*,L} - \mathfrak{p}^{*}}_2\leq \left(\expectation{\left[\sup_{x\in\mathcal{X}}|\expectation{[f(x,\zeta_L)]}-\expectation{[f(x,\zeta)]}|^2\right]}\right)^{1/2}}\leq c_1h_L^{\alpha}.\]

Now let us analyse $\displaystyle{\bnorm{\hat{\mathfrak{p}}^{*,L} - \mathfrak{p}^{*,L}}_2}$. Observe that,
\begin{align*}
    \bnorm{\hat{\mathfrak{p}}^{*,L} - \mathfrak{p}^{*,L}}_2 &= \bnorm{\min_{x \in \mathcal{X}}\left( \sum_{\ell = 0}^L \frac{1}{N_{\ell}}\sum_{k = 1}^{N_{\ell}}g(x,\bar{\zeta}_{\ell}^k)\right) - \min_{x\in \mathcal{}X}\sum_{\ell =0}^L \expectation{[g(x,\bar{\zeta}_{\ell})]}}_2\\
    &\leq \bnorm{\sup_{x\in \mathcal{X}}\left(\babs{\sum_{\ell = 0}^L \frac{1}{N_{\ell}}\sum_{k = 1}^{N_{\ell}}g(x,\bar{\zeta}_{\ell}^k) - \sum_{\ell =0}^L \expectation{[g(x,\bar{\zeta}_{\ell})]}}\right)}_2\\
    &\leq \bnorm{\sup_{x\in \mathcal{X}}\left(\sum_{\ell = 0}^L \babs{\frac{1}{N_{\ell}}\sum_{k = 1}^{N_{\ell}}(g(x,\bar{\zeta}_{\ell}^k) -\expectation{[g(x,\bar{\zeta}_{\ell})]})}\right)}_2\\
    &\leq \bnorm{\sum_{\ell = 0}^L\sup_{x\in \mathcal{X}}\babs{\frac{1}{N_{\ell}}\sum_{k = 1}^{N_{\ell}}(g(x,\bar{\zeta}_{\ell}^k) -\expectation{[g(x,\bar{\zeta}_{\ell})]})}}_2\\
    &\leq \sum_{\ell = 0}^L\bnorm{\sup_{x\in \mathcal{X}}\babs{\frac{1}{N_{\ell}}\sum_{k = 1}^{N_{\ell}}(g(x,\bar{\zeta}_{\ell}^k) -\expectation{[g(x,\bar{\zeta}_{\ell})]})}}_2.
\end{align*}
In order to study $\displaystyle{\bnorm{\sup_{x\in \mathcal{X}}\babs{\frac{1}{N_{\ell}}\sum_{k = 1}^{N_{\ell}}(g(x,\bar{\zeta}_{\ell}^k) -\expectation{[g(x,\bar{\zeta}_{\ell})]})}}_2}$, we extract tools from empirical process theory. For a given $\ell\geq0$, if we define,
\[
\mathbb{F}^{\ell}_{N_{\ell}}(\cdot) = \sqrt{N_{\ell}}\left(\frac{1}{N_{\ell}}\sum_{k = 1}^{N_{\ell}}(g(\cdot,\bar{\zeta}^{\ell}_k) - \expectation{[g(\cdot,\bar{\zeta}^{\ell})]})\right)
\]
then $F_{N_{\ell}}^{\ell}$ is an empirical process. Further, since $f(\cdot,\zeta_{\ell})$ is Lipschitz continuous, we conclude that is $g(\cdot,\bar{\zeta^{\ell}})$ is also Lipschitz. Thereby applying Lemma 3,4 and 5 and under the assumption \ref{assumption_mlmc_var}, we have 
\begin{align*}
   \sqrt{N_{\ell}} \bnorm{\sup_{x\in \mathcal{X}}\babs{\frac{1}{N_{\ell}}\sum_{k = 1}^{N_{\ell}}(g(x,\bar{\zeta}_{\ell}^k) -\expectation{[g(x,\bar{\zeta}_{\ell})]})}}_2 \leq C\norm{\bar{g}(x,\bar{\zeta}^{\ell})}_2\int_{0}^1\sqrt{1 + \log(Q_{[ ]}(v\norm{\bar{g}(x,\bar{\zeta}^{\ell})}_2,\mathfrak{F},\norm{\cdot}_2))}dv
\end{align*}
where $\bar{g}(x,\bar{\zeta}^{\ell}) = \sup_{x\in\mathcal{X}}|g(x,\bar{\zeta_{\ell}})-\expectation{[g(x,\bar{\zeta_{\ell}})]}|$. Referring to the calculations in \cite{lam2018bounding} 
 (Proposition EC.3), we have
\begin{align*}
    &C\norm{\bar{g}(x,\bar{\zeta}^{\ell})}_2\int_{0}^1\sqrt{1 + \log(Q_{[ ]}(v\norm{\bar{g}(x,\bar{\zeta}^{\ell})}_2,\mathfrak{F},\norm{\cdot}_2))}dv \leq\\ &C'\left(\norm{\bar{g}(x,\bar{\zeta}^{\ell})}_2 + \sqrt{d\log\left(\max\left\{ 3,\frac{12\mathcal{D}_{\mathcal{X}}L_{f}}{\norm{\bar{g}(x,\bar{\zeta}^{\ell})}_2}\right\}\right)}\min\left(4\mathcal{D}_{\mathcal{X}}L_{f},\norm{\bar{g}(x,\bar{\zeta}^{\ell})}_2\right)\right).
\end{align*}
Now if, $\displaystyle{\frac{12\mathcal{D}_{\mathcal{X}}L_{f}}{\norm{\bar{g}(x,\bar{\zeta}^{\ell})}_2} \leq 3}$,then we have, 
\begin{align*}
    C'\left(\norm{\bar{g}(x,\bar{\zeta}^{\ell})}_2 + \sqrt{d\log\left(\max\left\{ 3,\frac{12\mathcal{D}_{\mathcal{X}}L_{f}}{\norm{\bar{g}(x,\bar{\zeta}^{\ell})}_2}\right\}\right)}\min\left(4\mathcal{D}_{\mathcal{X}}L_{f},\norm{\bar{g}(x,\bar{\zeta}^{\ell})}_2\right)\right) \\\leq C'\norm{\bar{g}(x,\bar{\zeta}^{\ell})}_2(1+\sqrt{d\log(3)})
\end{align*}
Otherwise, we have
\begin{align*}
     C'\left(\norm{\bar{g}(x,\bar{\zeta}^{\ell})}_2 + \sqrt{d\log\left(\max\left\{ 3,\frac{12\mathcal{D}_{\mathcal{X}}L_{f}}{\norm{\bar{g}(x,\bar{\zeta}^{\ell})}_2}\right\}\right)}\min\left(4\mathcal{D}_{\mathcal{X}}L_{f},\norm{\bar{g}(x,\bar{\zeta}^{\ell})}_2\right)\right)\\
     \leq C'\norm{\bar{g}(x,\bar{\zeta}^{\ell})}_2 \left(1+\sqrt{d\log\left(\frac{12\mathcal{D}_{\mathcal{X}}L_{f}}{\norm{\bar{g}(x,\bar{\zeta}^{\ell})}_2}\right)}\right)\\
     \leq \mathfrak{c}_{\ell}(\norm{\bar{g}(x,\bar{\zeta}^{\ell})}_2)^{\frac{1}{1+a}}
\end{align*}
where the last inequality is the consequence of Lemma 7 for $0<a<\infty$ and for some constant $\mathfrak{c}_{\ell}$. Observe that as a consequence of assumption \ref{assumption_mlmc_var}, we have,
\begin{align*}
    \norm{\sup_{x\in\mathcal{X}}|g(x,\bar{\zeta_{\ell}})-\expectation{[g(x,\bar{\zeta_{\ell}})]}|}_2 &\leq \norm{\sup_{x \in \mathcal{X}}|f(x,\zeta_{\ell}) - f(x,\zeta_{\ell-1})|}_2 + \norm{\sup_{x \in \mathcal{X}}|\expectation{[f(x,\zeta_{\ell}) - f(x,\zeta_{\ell-1})]}|}_2\\
    &\leq 2c_2h_{\ell}^{\beta/2}
\end{align*}
Further, letting $\displaystyle{\bar{\mathfrak{c}} = \max_{0\leq\ell\leq L}\mathfrak{c}_{\ell}}$ and collating everything together we have
\[
\bnorm{\mathfrak{p}^{*,L} - \hat{\mathfrak{p}}^{*,L}}_2\leq 2c_2\bar{\mathfrak{c}}\sum_{\ell = 0}^{L}\frac{h_{\ell}^{\bar{\beta}/2}}{\sqrt{N_{\ell}}}
\]
where $\bar{\beta} = \beta\frac{1}{1+a}$. Putting everything together, we get the desired result.
\end{proof}
Based on the above formulation of the root-mean-squared error, we give below the result that shows the computational cost associated with the multilevel optimal gap estimator. 
\begin{corollary}
    Suppose assumptions \ref{lip_assumption},\ref{assumption_mlmc_bias} and \ref{assumption_mlmc_var} holds. Then, for any $\epsilon<\frac{1}{e}$ the computational complexity required for $\displaystyle{\norm{\hat{\mathfrak{p}}^{*,L} - \mathfrak{p}^*}_2}\leq \epsilon$ is
    \begin{equation}
    \label{eq:16}
        \mathcal{C}_{mlmc}^{\mathfrak{p}} = \begin{cases}
            \mathcal{O}(\epsilon^{-2}),&\text{ for } \bar{\beta} >1 .\\
            \mathcal{O}\left(\epsilon^{-2}\log(\epsilon^{-1})\right), &\text{ for } \bar{\beta} = 1.\\
           \mathcal{O}\left(\epsilon^{-\left(2+\frac{1-\bar{\beta}}{\alpha}\right)}\right), &\text{ for } \bar{\beta} < 1.\\
            
        \end{cases}
    \end{equation}
\end{corollary}
The proof of the above result follows the line of argument similar to the one observed by in \cite{giles2008multilevel,pages2018numerical} and is therefore skipped. 
\subsection{Optimal Gap Estimator}
Another aspect of SAA that is paramount among practitioners is the Optimal Gap estimator. As the name suggests, the primary aim of this estimator is to assess the quality of a candidate solution of the optimization problem \eqref{saa:eq1}. Mathematically, let $\hat{x}$ be a candidate solution. The quality of this solution is assessed using the optimality gap defined as
\begin{equation}
    \mathfrak{G}(\hat{x}):= F(\hat{x}) - \mathfrak{p}^*
\end{equation}
where $\displaystyle{\mathfrak{p}^* := \min_{x \in X}F(x)}$. The Monte Carlo approximation of $\mathfrak{G}(\hat{x})$ is given as,
\begin{equation}
    \hat{\mathfrak{G}}_N(\hat{x}): = \frac{1}{N}\sum_{k = 1}^Nf(\hat{x},\zeta_k) - \min_{x\in\mathcal{X}}\frac{1}{N}\sum_{k = 1}^Nf(x,\zeta_k),
\end{equation}
where $\{\zeta_k\}_{1\leq k \leq N}$ are i.i.d realisation of the random variable $\zeta$ that is common in both the terms in the above equation. In an unbiased realisation of the samples, the underlying mechanism is to statistically estimate the upper and lower bound of $\hat{\mathfrak{G}}_N(\hat{x})$ by performing $M$ independent estimation of $\hat{\mathfrak{G}}_N(\hat{x})$ and determining the one-sided confidence interval. This approach is well-documented and is readily used in various practical applications. Interested readers may refer to \cite{kleywegt2002sample, Mak1999Feb, Homem-de-Mello2014Jan} for a thorough discussion of this procedure. However, instead of estimating an upper and lower bound of the optimal value based on the optimal gap estimator, we undertake an RMSE analysis formulating the RMSE bound as a function of the bias parameter $h$ and the number of samples $N$. The following two propositions provide the RMSE bound in the context of Monte Carlo and MLMC paradigms, respectively. 

Let $\hat{x}$ be a candidate solution obtained by solving an SAA problem. Further let, $\displaystyle{\mathfrak{p}^{*,h} = \min_{x \in X}F^h(x)}$ and define,
$$
\mathfrak{G}^h(\hat{x}) := F_h(\hat{x}) - \mathfrak{p}^{*,h}
$$
$$\mathfrak{G}^h_N(\hat{x}) := \frac{1}{N}\sum_{k = 1}^N f(\hat{x},\zeta_h^k) - \min_{x \in \mathcal{X}}\frac{1}{N}\sum_{k = 1}^N f(x,\zeta_h^k)$$
The following result gives a root-mean-squared error formulation for the estimator defined above.
\begin{proposition}[Monte Carlo SAA]
\label{proposition:optimal_gap_mc}
    Suppose assumptions \ref{lip_assumption},\ref{bias_assmpution} and \ref{variance_assumption} holds. The for $N\in \mathbb{N}$ and $\hat{x} \in \mathcal{X}$,
\[
\norm{\mathfrak{G}^h_N(\hat{x}) - \mathfrak{G}(\hat{x})}_2 \leq 2c_1h^{\alpha} + \mathfrak{c}_2 \frac{\sigma}{\sqrt{N}} 
\]
\end{proposition}
\begin{proof}
   See Appendix \ref{appendix:B}.
\end{proof}
In the multilevel paradigm, the optimal gap estimator is defined as
\begin{align*}
\mathfrak{G}^L(\hat{x}) &= \expectation{[f(\hat{x},\zeta^L)]} - \min_{x \in \mathcal{X}}\expectation{[f(x,\zeta^L)]}\\
& = \sum_{\ell = 0}^L \expectation{[f(\hat{x},\zeta^{\ell}) - f(\hat{x},\zeta^{\ell-1})]} - \min_{x \in \mathcal{X}}\sum_{\ell = 0}^L\expectation{[f(x,\zeta^{\ell}) - f(x,\zeta^{\ell-1})]}
\end{align*}
with the Monte Carlo approximation being defined as,
\[
\hat{\mathfrak{G}}^L(\hat{x}) = \sum_{\ell = 0}^L \frac{1}{N_{\ell}}\sum_{k = 1}^{N_{\ell}}(f(\hat{x},\zeta_{\ell}^k)-f(\hat{x},\zeta_{\ell-1}^k)) - \min_{x \in \mathcal{X}}\left( \sum_{\ell = 0}^L \frac{1}{N_{\ell}}\sum_{k = 1}^{N_{\ell}}(f(x,\zeta_{\ell}^k)-f(x,\zeta_{\ell-1}^k))\right)  
\]
For the sake of notational convenience let $g(x,\bar{\zeta}_{\ell}):= f(x,\zeta_{\ell}) - f(x,\zeta_{\ell-1})$, $\mathfrak{p}^{*,L}: = \min_{x\in \mathcal{}X}\sum_{\ell =0}^L \expectation{[g(x,\bar{\zeta}_{\ell})]}$ and let $\hat{\mathfrak{p}}^{*,L}$ be its monte carlo approximation. Then,
\[
\hat{\mathfrak{G}}^L(\hat{x}) = \sum_{\ell = 0}^L \frac{1}{N_{\ell}}\sum_{k = 1}^{N_{\ell}}g(\hat{x},\bar{\zeta}^k_{\ell}) - \min_{x \in \mathcal{X}}\sum_{\ell = 0}^L \frac{1}{N_{\ell}}\sum_{k = 1}^{N_{\ell}}g(x,\bar{\zeta}^k_{\ell})
\]
\begin{proposition}[MLMC-SAA]
\label{proposition:optimal_gap_mlmc}
    Suppose assumptions \ref{lip_assumption},\ref{assumption_mlmc_bias} and \ref{assumption_mlmc_var} holds. Then for any $0<a<\infty$, $L\geq2$ and any $\hat{x}\in\mathcal{X}$,
    \[
    \norm{\hat{\mathfrak{G}}^L(\hat{x}) - \mathfrak{G}(\hat{x})}_2 \leq  2c_1h_{L}^{\alpha} +\mathfrak{c}_3 \sum_{\ell = 0}^L\frac{ h_{\ell}^{\bar{\beta}/2}}{\sqrt{N_{\ell}}}
    \]
    where $\bar{\beta} = \beta\frac{1}{1+a}$.
\end{proposition}
\begin{proof}
See Appendix \ref{appendix:B}.
\end{proof}

\section{Numerical Illustration}
\label{section:numerical_illustration}
In this section, we undertake numerical experimentation to illustrate the impact of biased approximation of $\zeta$ on sample average approximation. To this end, we consider two minimization problems in the coherent risk measure paradigms,\textit{i.e.}, Conditional Value at Risk or CVaR.  Recall that,
\begin{align*}
    \text{CVaR}_{\theta}(\zeta) = \min_{x \in \mathbb{R}}\left\{F(x) = \expectation{\left[x + \frac{1}{1-\theta}(\zeta - x)_{+}\right]}\right\}
\end{align*}
where $\theta$ is the confidence level. For our first example, we consider a portfolio consisting of a single put option where the asset price is driven by a geometric Brownian Motion(gBm) given as,
\begin{equation}
    dX_t = rX_tdt + \sigma X_tdW_t
\end{equation}
where $r$  and $\sigma$ denote the risk-free rate of return and the volatility, and $W_t$ denotes the standard Brownian Motion. As for the second example, we look into the scenario simulation paradigm developed by in \cite{gordy2010nested} leading to a nested simulation framework. Before we dwell on the numerical simulation, we make certain observations on the function ${f(x,\zeta):=x + \frac{1}{1-\theta}(\zeta - x)_{+}} $, and discuss the underlying algorithm in order to undergo numerical experimentation. To begin with, let $\zeta_{\ell}$ and $\zeta_{\ell'}$ be two approximation of the random variable $\zeta$, then,
\begin{equation}
    \abs{f(x,\zeta_{\ell}) - f(x,\zeta_{\ell'})}\leq \frac{1}{1-\theta}\abs{\zeta_{\ell}-\zeta_{\ell'}}
\end{equation}
for all $x \in \mathcal{X}$. Therefore, 
\begin{equation}
    \sup_{x \in \mathcal{X}}\abs{f(x,\zeta_{\ell}) - f(x,\zeta_{\ell'})}\leq \frac{1}{1-\theta}\abs{\zeta_{\ell}-\zeta_{\ell'}}.
\end{equation}
Consequently, we have,
\begin{align}
    \sup_{x\in\mathcal{X}}\expectation{\left[ \abs{f(x,\zeta_{\ell}) - f(x,\zeta_{\ell'})}\right]}\leq \frac{1}{1-\theta}\expectation{\abs{\zeta_{\ell}-\zeta_{\ell'}}}\\
    \expectation{\left[\sup_{x\in\mathcal{X}}\abs{f(x,\zeta_{\ell}) - f(x,\zeta_{\ell'})}^2\right]}\leq \frac{1}{(1-\theta)^2}\expectation{[\abs{\zeta_{\ell}-\zeta_{\ell'}}^2]}.
\end{align}
The above two inequalities help us determine $\alpha$ and $\beta$, \textit{i.e.} the bias and variance convergence rate, essential for our multilevel simulation. The next step in our simulation is determining the number of samples to achieve a root-mean-squared error of $\epsilon$. In the multilevel paradigm, the theoretical analysis undertaken in the previous section suggests that taking,
\[
N_{\ell} = \left \lceil\frac{16}{\epsilon^2}(c_2\bar{\mathfrak{c}})^2 h_{\ell}^{\frac{\Bar{\beta}+2}{3}}\left(\sum_{\ell = 0}^L h_{\ell}^{\frac{\Bar{\beta}-1}{3}}\right)^2\right \rceil \text{ and } L = \left\lceil\frac{\log(2c_1h_0^{\alpha}\epsilon^{-1})}{\alpha\log(m)}\right\rceil
\]
would lead to root-mean-squared error of $\mathcal{O}(\epsilon)$. However, the formulation requires us to estimate various constants, which, given one is able to calculate, leads to very conservative numbers for $N_{\ell}$, whereas, in the application, we often do not have the computational budget to perform simulations based on these estimates. Then, how does one go about reaping the benefits of the multilevel approximation in the SAA paradigm? We resolve this question by roughly approximating the optimal solution by solving an SAA problem with very few samples, subsequently generating more samples based on this estimate for an accurate approximation. The idea is not particularly novel, with a similar approach being studied by in \cite{pasupathy2010choosing}, where they intend to progressively increase the sample size based on some performance criteria of the previous estimation, discussion about which is beyond the scope of this paper. Here, we consider a version of the procedure where instead of progressively increasing the sample size, we estimate the number of samples based on the formulas presented in Chapter 9 of \cite{pages2018numerical} both in the context of MLMC and Monte Carlo estimation. The formula and algorithm, both in the context of Monte Carlo SAA and MLMC-SAA, are presented in Appendix \ref{appendix:A}.

\subsection{Geometric Brownian Motion}
For our first example, we shall consider an investment consisting of a short position in a single put option, where the loss is defined as
\[
\zeta := (K-X_T)_+ - e^{rT}P_0,
\]
with $P_0$ being the initial price at which the option was sold. We assume the underlying stock $X_t$ follows gBm, \textit{i.e.},
\[
 dX_t = rX_tdt + \sigma X_tdW_t
\]
and further assume $X_0 = 100$, $r = 0.05$, $\sigma = 0.2$, $T = 1$, $P_0 = 10.7$ and (strike price) $K = 110$ for our simulation \cite{bardou2009computing}. In order to undergo our simulation, we discretize the gBm using Euler-Maruyama and Milstein numerical scheme, given as,
\begin{align}
X_{n+1} = X_n + rX_nh + \sigma X_n \Delta W_n  ~~~~~\text{(Euler-Marumaya)}\\
X_{n+1} = X_n + (r-\frac{1}{2}\sigma^2)hX_n + \sigma X_n \Delta W_n + \frac{1}{2}\sigma^2X_n(\Delta W_n)^2~~~~~~\text{(Milstein)}, 
\end{align}
where $h = T/m$ is the step size and $\Delta W_n = W_{n+1} - W_{n}$. Here, we take $m = 4$ as the refinement factor $h_0 = T$ for the purpose of our simulation. The value of $\alpha$ and $\beta$ can be estimated based on equation (27) and (28). Referring to the analysis in \cite{giles2019analysis}, we have, $\alpha = 1$ and $\beta = 1$ for Euler-Maruyama scheme and $\alpha = 1$ and $\beta = 2$ for Milstein Scheme. Also, we take $a = 10^{-3}$. Based on these parameters, we undertake our simulation where we perform a $100$ independent run of algorithm \ref{algo:mc_saa} and \ref{algo:mlmc_saa} (see Appendix \ref{appendix:A}), performing minimization over the interval $\mathcal{X} = [23,25]$ for $\theta = 0.95$. Finally, we estimate the optimal value as the average of all the independent simulations. We also estimate the RMSE and the $\probability{(\abs{\hat{\mathfrak{p}}^{*,\cdot} - \mathfrak{p}^*}>\epsilon)}$, for a given $\epsilon$. For $\theta = 0.95$, the value of CVaR($\zeta$) is approximately equal to $30.347$. The results from these simulations are tabulated in Table \ref{tab:mlmc_eu},\ref{tab:mc_eu},\ref{tab:mc_mil} and \ref{tab:mlmc_mil}. Figure \ref{fig:rmse_sde} depicts the graphical representation of the tabulated results where we observe the computational saving achieved by MLMC-SAA.

\begin{table}[H]
\begin{tabular}{|c|c|c|c|c|c|c|c|}
\hline
$\epsilon$ & $h_0$   & Bias       & Variance   & RMSE   & $\probability{(\abs{\hat{\mathfrak{p}}^{*,\cdot} - \mathfrak{p}^*}>\epsilon)}$ & Cost       & Value  \\ \hline
0.5000     & 1.00000 & 2.7188e-01 & 3.0176e-01 & 0.6129 & 0.420                                                                          & 1.8381e+05 & 30.619 \\ \hline
0.2500     & 1.00000 & 5.5760e-02 & 7.9016e-02 & 0.2866 & 0.420                                                                          & 1.2872e+06 & 30.403 \\ \hline
0.1250     & 1.00000 & 4.9798e-02 & 2.0944e-02 & 0.1530 & 0.420                                                                          & 9.2293e+06 & 30.397 \\ \hline
0.0625     & 1.00000 & 1.1595e-02 & 5.8147e-03 & 0.0771 & 0.440                                                                          & 7.5802e+07 & 30.359 \\ \hline
0.0312     & 1.00000 & 2.2947e-02 & 1.3838e-03 & 0.0437 & 0.420                                                                          & 4.8290e+08 & 30.370 \\ \hline
\end{tabular}
\caption{MLMC SAA estimation of CVaR - Euler Maruyama Scheme}
\label{tab:mlmc_eu}
\end{table}

\begin{table}[H]
\begin{tabular}{|c|c|c|c|c|c|c|c|}
\hline
$\epsilon$ & $h_0$   & Bias       & Variance   & RMSE   & $\probability{(\abs{\hat{\mathfrak{p}}^{*,\cdot} - \mathfrak{p}^*}>\epsilon)}$ & Cost       & Value  \\ \hline
0.5000     & 0.50000 & 2.3638e+00 & 5.0417e-01 & 2.4682 & 0.990                                                                          & 1.0498e+04 & 32.711 \\ \hline
0.2500     & 0.25000 & 1.1381e+00 & 5.8845e-02 & 1.1637 & 1.000                                                                          & 8.7483e+04 & 31.485 \\ \hline
0.1250     & 0.12500 & 5.4815e-01 & 1.1016e-02 & 0.5581 & 1.000                                                                          & 6.6049e+05 & 30.895 \\ \hline
0.0625     & 0.06250 & 2.8821e-01 & 2.9261e-03 & 0.2932 & 1.000                                                                          & 5.6448e+06 & 30.635 \\ \hline
0.0312     & 0.03125 & 1.5071e-01 & 7.7410e-04 & 0.1533 & 1.000                                                                          & 4.4189e+07 & 30.498 \\ \hline
\end{tabular}
\caption{Monte Carlo SAA estimation of CVaR - Euler Maruyama Scheme}
\label{tab:mc_eu}
\end{table}

\begin{table}[H]
\begin{tabular}{|c|c|c|c|c|c|c|c|}
\hline
$\epsilon$ & $h_0$   & Bias       & Variance   & RMSE   & $\probability{(\abs{\hat{\mathfrak{p}}^{*,\cdot} - \mathfrak{p}^*}>\epsilon)}$ & Cost       & Value  \\ \hline
0.5000     & 0.50000 & 9.4734e-01 & 1.3697e-01 & 1.0171 & 0.900                                                                          & 1.7462e+04 & 29.400 \\ \hline
0.2500     & 0.25000 & 3.8208e-01 & 2.5882e-02 & 0.4146 & 0.820                                                                          & 1.3750e+05 & 29.965 \\ \hline
0.1250     & 0.12500 & 1.8202e-01 & 7.8050e-03 & 0.2023 & 0.690                                                                          & 1.0405e+06 & 30.165 \\ \hline
0.0625     & 0.06250 & 8.2207e-02 & 1.9756e-03 & 0.0935 & 0.690                                                                          & 8.4570e+06 & 30.265 \\ \hline
0.0312     & 0.03125 & 3.9240e-02 & 5.6320e-04 & 0.0459 & 0.610                                                                          & 7.3093e+07 & 30.308 \\ \hline
\end{tabular}
\caption{Monte Carlo SAA estimation of CVaR - Milstein Scheme}
\label{tab:mc_mil}
\end{table}

\begin{table}[H]
\begin{tabular}{|c|c|c|c|c|c|c|c|}
\hline
$\epsilon$ & $h_0$   & Bias       & Variance   & RMSE   & $\probability{(\abs{\hat{\mathfrak{p}}^{*,\cdot} - \mathfrak{p}^*}>\epsilon)}$ & Cost       & Value  \\ \hline
0.5000     & 1.00000 & 1.7175e-01 & 1.5618e-01 & 0.4309 & 0.280                                                                          & 2.0223e+04 & 30.175 \\ \hline
0.2500     & 1.00000 & 7.9033e-02 & 4.3766e-02 & 0.2236 & 0.310                                                                          & 8.7559e+04 & 30.268 \\ \hline
0.1250     & 1.00000 & 3.2663e-02 & 8.4547e-03 & 0.0976 & 0.230                                                                          & 6.1980e+05 & 30.314 \\ \hline
0.0625     & 1.00000 & 1.1434e-02 & 2.4114e-03 & 0.0504 & 0.200                                                                          & 2.6256e+06 & 30.336 \\ \hline
0.0312     & 1.00000 & 3.1578e-04 & 5.4172e-04 & 0.0233 & 0.200                                                                          & 1.9655e+07 & 30.347 \\ \hline
\end{tabular}
\caption{MLMC SAA estimation of CVaR - Milstein Scheme}
\label{tab:mlmc_mil}
\end{table}
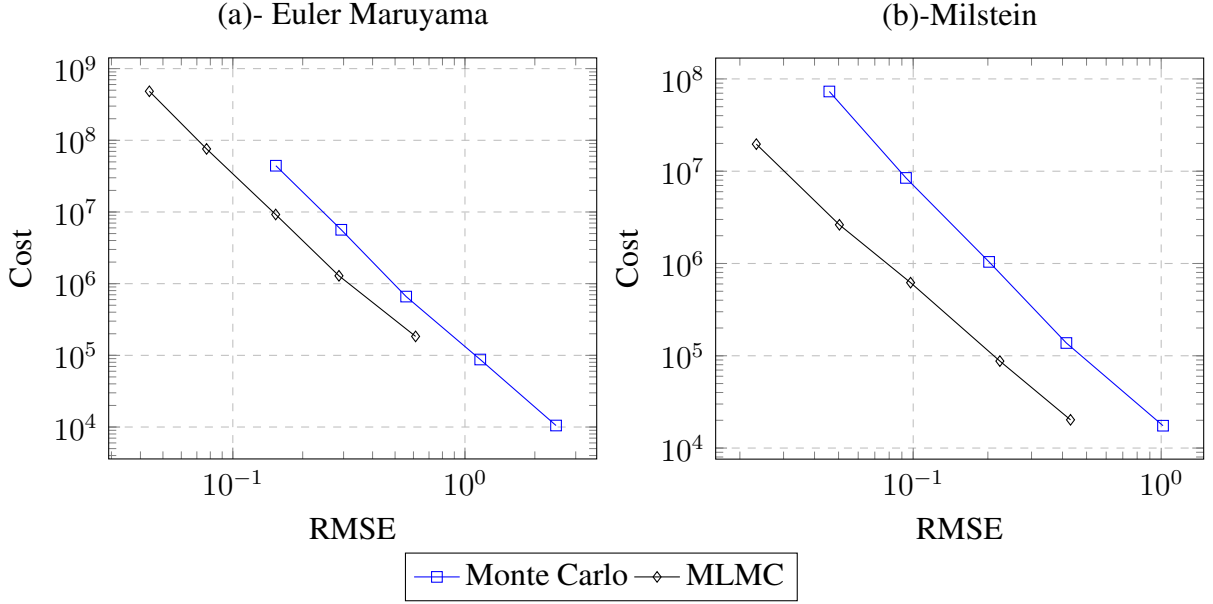
\begin{figure}
\pgfplotstableread{mc_saa_sde_eu.dat}{\mcrmse}
\pgfplotstableread{mlmc_saa_sde_eu.dat}{\mlmcrmse}
\pgfplotstableread{mc_saa_sde_ml.dat}{\mcrmseml}
\pgfplotstableread{mlmc_saa_sde_ml.dat}{\mlmcrmseml}  
\begin{center}
\begin{tikzpicture}[scale=1]
\begin{loglogaxis}[
legend columns=-1,
legend entries={Monte Carlo, MLMC},
legend to name=named,
minor tick num=1,
xlabel= {RMSE},ylabel = {Cost},
    ymajorgrids=true,
    xmajorgrids=true,
    grid style=dashed,
    title = {(a)- Euler Maruyama}]
\addplot [color = blue,mark  = square] table [x={RMSE}, y expr=\thisrow{Cost}] {\mcrmse};
\addplot [color = black,mark = diamond] table [x={RMSE}, y expr=\thisrow{Cost}] {\mlmcrmse};
\end{loglogaxis}
\end{tikzpicture}
\begin{tikzpicture}[scale=1]
\begin{loglogaxis}[
minor tick num=1,
xlabel= {RMSE},ylabel = {Cost},
    ymajorgrids=true,
    xmajorgrids=true,
    grid style=dashed,title = {(b)-Milstein}]
\addplot [color = blue,mark  = square] table [x={RMSE}, y expr=\thisrow{Cost}] {\mcrmseml};
\addplot [color = black,mark = diamond] table [x={RMSE}, y expr=\thisrow{Cost}] {\mlmcrmseml};
\end{loglogaxis}
\end{tikzpicture}
\ref{named}
\end{center}
\caption{(a) Euler-Maruyama Approximation- Computational Cost as a function of RMSE log-log scale. (b) Milstein Approximation- Computational Cost as function of RMSE log-log sclae}
\label{fig:rmse_sde}
\end{figure}

\subsection{Nested Simulation}
For the second example, we refer to the research carried out in \cite{gordy2010nested,giles2019multilevel}, where the authors formulated the estimation of CVaR as a nested expectation problem. Consequently, for our simulation, we define $\zeta$ as follows,
\begin{equation}
    \zeta:= -1 - \expectation{[\phi(Y,Z)|Y]}
\end{equation}
where, $\displaystyle{\phi(y,z): = -\tau y^2 -2\sqrt{\tau(1-\tau)}yz - (1-\tau)z^2}$ and $y,z\in\mathbb{R}$. Also, $Y,Z$ are independent following normal distribution $\mathcal{N}(0,1)$. The above formulation considers an option with payoff $-W_T^2$ at time $T = 1$. The value of the option at a time $t$ is given by  $\displaystyle{P(t,y): = \expectation[-W_T^2|W_t = y]}$ and the loss  $\zeta$ is given as, $\zeta := P(0,0) - P(\tau,W_{\tau})$, where $\tau \in (0,1)$ is the time horizon.
For further information on the above formulation and the analytical calculations, refer to \cite{giles2019multilevel,crepey2023multilevel}. For now, observe that the above formulation leads to a conditional optimization problem. Therefore, the Monte Carlo approximation of $F(x)$ is based on generating inner and outer samples. Let the bias parameter $h = 1/M$ where $M$ is the number of inner samples and let $N$ denote the number of outer samples; then,
\begin{equation}
F^h_N(x) = \frac{1}{N}\sum_{k = 1}^N\left(x + \frac{1}{1-\theta}\left(-1-\frac{1}{M}\sum_{j = 1}^M\phi(Z_j,Y_{k})-x\right)_+\right)
\end{equation}
gives the Monte Carlo approximation of $F(x)$. We take $\tau = 0.5$ and $\theta = 0.975$ for our numerical simulation. For the practical implementation, we minimize over the interval $[1,4]$. For the multilevel approximation, let
\[
\hat{E}_{M_{\ell}}(Y_k) = -1-\frac{1}{M_{\ell}}\sum_{j = 1}^{M_{\ell}}\phi(Z_j,Y_{k})
\] and further define,
\[
f(x,\hat{E}_{M_{\ell}}(Y_k)):= x + \frac{1}{1-\theta}\left(\hat{E}_{M_{\ell}}(Y_k)-x\right)_+ .
\] Then $\hat{F}_{L}(x)$, is given as
\begin{equation}
\hat{F}_{L}(x):=\sum_{\ell = 0}^L \frac{1}{N_{\ell}}\sum_{k = 1}^{N_{\ell}}\left(f(x,\hat{E}_{M_{\ell}}(Y_k)) - f(x,\hat{E}_{M_{\ell-1}}(Y_k))\right).
\end{equation}
To this end, observe that for a given $Y$, we have,
\begin{align*}
    \sup_{x \in [1,4]}\babs{f(x,\hat{E}_{M_{\ell}}(Y))- f(x,\hat{E}_{M_{\ell-1}}(Y_k))}^2 \leq \frac{1}{(1-\theta)^{2}}\babs{\hat{E}_{M_{\ell}}(Y) - \hat{E}_{M_{\ell-1}}(Y)}^2
    \end{align*}
where the right-hand side of the above inequality is independent of $x$. Therefore, we have
\begin{align*}
    \expectation{ \left[\sup_{x \in [1,4]}\babs{f(x,\hat{E}_{M_{\ell}}(Y))- f(x,\hat{E}_{M_{\ell-1}}(Y_k))}^2\right]}\leq \frac{1}{(1-\theta)^{2}}\expectation{\left[\babs{\hat{E}_{M_{\ell}}(Y) - \hat{E}_{M_{\ell-1}}(Y)}^2\right]}.
\end{align*}
Now, as a consequence of Proposition 9.2 (a) in \cite{pages2018numerical}, Proposition 3 in \cite{gordy2010nested} and equation (27) and (28), we have $\beta = 1$ and $\alpha = 1$ albeit under some regularity assumptions. For multilevel simulation, we take $h_0 = 1/64$, \textit{i.e.}, $M = 64$ and take $a = 10^{-3}$. As before, we perform $100$ independent simulation and estimate the RMSE and $\probability{(\abs{\hat{\mathfrak{p}}^{*,\cdot} - \mathfrak{p}^*}>\epsilon)}$. Table \ref{tab:mlmc_nes} and \ref{tab:mc_nes} tabulate the results obtained through MLMC-SAA and Monte Carlo SAA, respectively. Finally, in figure \ref{fig:rmse_nested}, we provide a graphical representation depicting the computational saving achieved by MLMC-SAA in the nested simulation framework.
\begin{table}[H]
\begin{tabular}{|c|c|c|c|c|c|c|c|}
\hline
$\epsilon$ & $h_0$   & Bias       & Variance   & RMSE   & $\probability{(\abs{\hat{\mathfrak{p}}^{*,\cdot} - \mathfrak{p}^*}>\epsilon)}$ & Cost       & Value \\ \hline
0.5000     & 0.01562 & 3.3295e-02 & 1.5839e-01 & 0.3994 & 0.210                                                                          & 4.6912e+04 & 2.877 \\ \hline
0.2500     & 0.01562 & 7.5835e-03 & 5.0063e-02 & 0.2239 & 0.290                                                                          & 1.5027e+05 & 2.918 \\ \hline
0.1250     & 0.01562 & 9.4439e-03 & 1.1780e-02 & 0.1089 & 0.210                                                                          & 5.9360e+05 & 2.901 \\ \hline
0.0625     & 0.01562 & 2.2478e-03 & 3.1485e-03 & 0.0562 & 0.260                                                                          & 3.7132e+06 & 2.912 \\ \hline
0.0312     & 0.01562 & 4.3228e-04 & 6.5439e-04 & 0.0256 & 0.250                                                                          & 2.0315e+07 & 2.910 \\ \hline
\end{tabular}
\caption{MLMC SAA estimation of CVaR - Nested}
\label{tab:mlmc_nes}
\end{table}
\begin{table}[H]
\begin{tabular}{|c|c|c|c|c|c|c|c|}
\hline
$\epsilon$ & $h_0$   & Bias       & Variance   & RMSE   & $\probability{(\abs{\hat{\mathfrak{p}}^{*,\cdot} - \mathfrak{p}^*}>\epsilon)}$ & Cost       & Value \\ \hline
0.5000     & 0.01562 & 1.7884e-02 & 2.5430e-01 & 0.5046 & 0.320                                                                          & 2.1908e+04 & 2.892 \\ \hline
0.2500     & 0.00781 & 1.1132e-02 & 6.4929e-02 & 0.2551 & 0.290                                                                          & 1.6616e+05 & 2.899 \\ \hline
0.1250     & 0.00391 & 2.6002e-03 & 1.9075e-02 & 0.1381 & 0.320                                                                          & 1.2272e+06 & 2.913 \\ \hline
0.0625     & 0.00195 & 1.3558e-03 & 4.6695e-03 & 0.0683 & 0.350                                                                          & 9.5507e+06 & 2.911 \\ \hline
0.0312     & 0.00098 & 9.7929e-03 & 1.3391e-03 & 0.0379 & 0.390                                                                          & 8.1160e+07 & 2.900 \\ \hline
\end{tabular}
\caption{Monte Carlo SAA estimation of CVaR - Nested}
\label{tab:mc_nes}
\end{table}
\pgfplotstableread{mc_saa_nested.dat}{\mcrmsenes}
\pgfplotstableread{mlmc_saa_nested.dat}{\mlmcrmsenes}
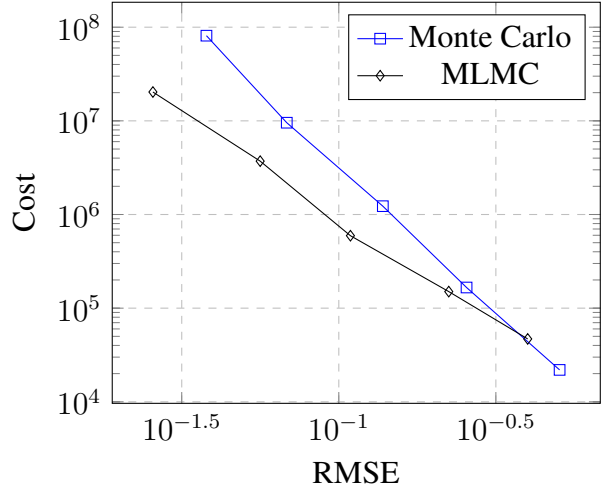
\begin{figure}
\begin{center}

\begin{tikzpicture}[scale=1]
\begin{loglogaxis}[minor tick num=1,
xlabel= {RMSE},ylabel = {$\text{Cost}$}, legend pos=north east,
    ymajorgrids=true,
    xmajorgrids=true,
    grid style=dashed]
\addplot [color = blue,mark  = square] table [x={RMSE}, y expr=\thisrow{Cost}] {\mcrmsenes};
\addplot [color = black,mark = diamond] table [x={RMSE}, y expr=\thisrow{Cost}] {\mlmcrmsenes};
 \legend{Monte Carlo,MLMC}
\end{loglogaxis}
\end{tikzpicture}
\end{center}
\caption{Nested Simulation-Computational Cost as a function of RMSE. log-log scale.}
\label{fig:rmse_nested}
\end{figure}

\section{Discussion and Conclusion}
In this paper, we looked into a discussion about the Sample average approximation problem,  where the random variable $\zeta$ is sampled from an approximate distribution, introducing bias in the Monte Carlo estimation of the expectation. We extended the conventional SAA setup to the multilevel framework to enhance the computational cost associated with performing the optimization procedure. Following the traditional analysis, we undertook the uniform convergence analysis, establishing the rate of convergence and the sampling complexity results both in the standard and multilevel context. We further analysed RMSE and derived the sample complexity results to achieve $\epsilon-$RMSE. Finally, we demonstrated the benefits of incorporating MLMC via a series of numerical examples. 

Although the results in the section \ref{section:numerical_illustration} do illustrate the advantages of MLMC incorporation, the underlying algorithm presented (Appendix \ref{appendix:A}) is heuristic and does require an extensive study. In this regard, one can look into the Retrospective Approximation studied in \cite{pasupathy2010choosing}, developing a more rigorous mathematical foundation for our algorithm. Also, the presented study provides a framework for further enhancements of the procedure in future works. For instance, Multilevel Richardson-Romberg extrapolation (ML2R), developed in \cite{pages2007multi} can be incorporated to improve the computational complexity associated with $\beta(\bar{\beta})\leq 1$. One can further enhance the cost by incorporating variance reduction techniques in the studied framework, such as Antithetic sampling \cite{giles2013antithetic}, importance sampling \cite{kebaier2018coupling,alaya2016improved}. These are various topics that we will explore in future research.

\bibliographystyle{elsarticle-num}
\bibliography{reference.bib}

\begin{thebibliography}{10}
\expandafter\ifx\csname url\endcsname\relax
  \def\url#1{\texttt{#1}}\fi
\expandafter\ifx\csname urlprefix\endcsname\relax\def\urlprefix{URL }\fi
\expandafter\ifx\csname href\endcsname\relax
  \def\href#1#2{#2} \def\path#1{#1}\fi

\bibitem{kleywegt2002sample}
A.~J. Kleywegt, A.~Shapiro, T.~Homem-de Mello, The sample average approximation
  method for stochastic discrete optimization, SIAM Journal on optimization
  12~(2) (2002) 479--502.

\bibitem{shapiro2021lectures}
A.~Shapiro, D.~Dentcheva, A.~Ruszczynski, Lectures on stochastic programming:
  modeling and theory, SIAM, 2021.

\bibitem{hong2017kernel}
L.~J. Hong, S.~Juneja, G.~Liu, Kernel smoothing for nested estimation with
  application to portfolio risk measurement, Operations Research 65~(3) (2017)
  657--673.

\bibitem{hu2020sample}
Y.~Hu, X.~Chen, N.~He, Sample complexity of sample average approximation for
  conditional stochastic optimization, SIAM Journal on Optimization 30~(3)
  (2020) 2103--2133.

\bibitem{wang2017stochastic}
M.~Wang, E.~X. Fang, H.~Liu, Stochastic compositional gradient descent:
  algorithms for minimizing compositions of expected-value functions,
  Mathematical Programming 161 (2017) 419--449.

\bibitem{yermol1971general}
Y.~M. Yermol'yev, A general stochastic programming problem (1971).

\bibitem{ermoliev2013sample}
Y.~M. Ermoliev, V.~I. Norkin, Sample average approximation method for compound
  stochastic optimization problems, SIAM Journal on Optimization 23~(4) (2013)
  2231--2263.

\bibitem{dentcheva2017statistical}
D.~Dentcheva, S.~Penev, A.~Ruszczy{\'n}ski, Statistical estimation of composite
  risk functionals and risk optimization problems, Annals of the Institute of
  Statistical Mathematics 69 (2017) 737--760.

\bibitem{rockafellar2000optimization}
R.~T. Rockafellar, S.~Uryasev, et~al., Optimization of conditional
  value-at-risk, Journal of risk 2 (2000) 21--42.

\bibitem{ruszczynski2006optimization}
A.~Ruszczy{\'n}ski, A.~Shapiro, Optimization of convex risk functions,
  Mathematics of operations research 31~(3) (2006) 433--452.

\bibitem{gordy2010nested}
M.~B. Gordy, S.~Juneja, Nested simulation in portfolio risk measurement,
  Management Science 56~(10) (2010) 1833--1848.

\bibitem{broadie2011efficient}
M.~Broadie, Y.~Du, C.~C. Moallemi, Efficient risk estimation via nested
  sequential simulation, Management Science 57~(6) (2011) 1172--1194.

\bibitem{giles2008multilevel}
M.~B. Giles, Multilevel monte carlo path simulation, Operations research 56~(3)
  (2008) 607--617.

\bibitem{giles2008improved}
M.~Giles, Improved multilevel monte carlo convergence using the milstein
  scheme, in: Monte Carlo and Quasi-Monte Carlo Methods 2006, Springer, 2008,
  pp. 343--358.

\bibitem{giles2013antithetic}
M.~B. Giles, L.~Szpruch, Antithetic multilevel monte carlo estimation for
  multidimensional sdes, in: Monte Carlo and Quasi-Monte Carlo Methods 2012,
  Springer, 2013, pp. 367--384.

\bibitem{pages2018numerical}
G.~Pag{\`e}s, Numerical probability, Universitext, Springer (2018).

\bibitem{ruszczynski2003stochastic}
A.~Ruszczy{\'n}ski, A.~Shapiro, Stochastic programming models, Handbooks in
  operations research and management science 10 (2003) 1--64.

\bibitem{wellner2013weak}
J.~Wellner, et~al., Weak convergence and empirical processes: with applications
  to statistics, Springer Science \& Business Media, 2013.

\bibitem{lam2018bounding}
H.~Lam, H.~Qian, Bounding optimality gap in stochastic optimization via
  bagging: Statistical efficiency and stability, arXiv preprint
  arXiv:1810.02905 (2018).

\bibitem{heinrich2001multilevel}
S.~Heinrich, Multilevel monte carlo methods, in: International Conference on
  Large-Scale Scientific Computing, Springer, 2001, pp. 58--67.

\bibitem{giles2019multilevel}
M.~B. Giles, A.-L. Haji-Ali, Multilevel nested simulation for efficient risk
  estimation, SIAM/ASA Journal on Uncertainty Quantification 7~(2) (2019)
  497--525.

\bibitem{haji2022adaptive}
A.-L. Haji-Ali, J.~Spence, A.~L. Teckentrup, Adaptive multilevel monte carlo
  for probabilities, SIAM Journal on Numerical Analysis 60~(4) (2022)
  2125--2149.

\bibitem{dereich2021general}
S.~Dereich, General multilevel adaptations for stochastic approximation
  algorithms ii: Clts, Stochastic Processes and their Applications 132 (2021)
  226--260.

\bibitem{frikha2016multi}
N.~Frikha, Multi-level stochastic approximation algorithms (2016).

\bibitem{crepey2023multilevel}
S.~Cr{\'e}pey, N.~Frikha, A.~Louzi, A multilevel stochastic approximation
  algorithm for value-at-risk and expected shortfall estimation, arXiv preprint
  arXiv:2304.01207 (2023).

\bibitem{banholzer2022rates}
D.~Banholzer, J.~Fliege, R.~Werner, On rates of convergence for sample average
  approximations in the almost sure sense and in mean, Mathematical Programming
   1--39.

\bibitem{Mak1999Feb}
W.~K. Mak, D.~P. Morton, R.~K. Wood, {Monte Carlo bounding techniques for
  determining solution quality in stochastic programs}, Oper. Res. Lett. 24~(1)
  (1999) 47--56.
\newblock \href {https://doi.org/10.1016/S0167-6377(98)00054-6}
  {\path{doi:10.1016/S0167-6377(98)00054-6}}.

\bibitem{Homem-de-Mello2014Jan}
T.~Homem-de Mello, G.~Bayraksan, {Monte Carlo sampling-based methods for
  stochastic optimization}, Surveys in Operations Research and Management
  Science 19~(1) (2014) 56--85.
\newblock \href {https://doi.org/10.1016/j.sorms.2014.05.001}
  {\path{doi:10.1016/j.sorms.2014.05.001}}.

\bibitem{pasupathy2010choosing}
R.~Pasupathy, On choosing parameters in retrospective-approximation algorithms
  for stochastic root finding and simulation optimization, Operations Research
  58~(4-part-1) (2010) 889--901.

\bibitem{bardou2009computing}
O.~Bardou, N.~Frikha, G.~Pages, Computing var and cvar using stochastic
  approximation and adaptive unconstrained importance sampling (2009).

\bibitem{giles2019analysis}
M.~B. Giles, K.~Debrabant, A.~R{\"o}ssler, Analysis of multilevel monte carlo
  path simulation using the milstein discretisation, Discrete \& Continuous
  Dynamical Systems-B 24~(8) (2019) 3881.

\bibitem{pages2007multi}
G.~Pag{\`e}s, Multi-step richardson-romberg extrapolation: remarks on variance
  control and complexity (2007).

\bibitem{kebaier2018coupling}
A.~Kebaier, J.~Lelong, Coupling importance sampling and multilevel monte carlo
  using sample average approximation, Methodology and Computing in Applied
  Probability 20~(2) (2018) 611--641.

\bibitem{alaya2016improved}
M.~B. Alaya, K.~Hajji, A.~Kebaier, Improved adaptive multilevel monte carlo and
  applications to finance, arXiv preprint arXiv:1603.02959 (2016).

\end{thebibliography}
\newpage
\begin{appendices}
\section{Appendix}
\label{appendix:A}
\begin{algorithm}
    \caption{Monte Carlo SAA}
    \label{algo:mc_saa}
    \textbf{Input}: $\epsilon,\hat{N}_0,\alpha$\\
    \textit{Step 1}: $\displaystyle{h_0 = \mathcal{O}(\epsilon^{\frac{1}{\alpha}})}$\\ 
    \textit{Step 2}: Generate $\zeta^1_{h_0},\zeta^2_{h_0},\dots,\zeta^{\hat{N}_0}_{h_0}$ independent and identically distributed sample of the random variable $\zeta_{h_0}$.\\
    \textit{Step 3}: Estimate $\displaystyle{\hat{x} = \arg\min_{x\in \mathcal{X}}\frac{1}{\hat{N}_0}\sum_{k = 1}^{\hat{N}_0} f(x,\zeta^{k}_{h_0})}$.\\
    \textit{Step 4}: Estimate $\displaystyle{\sigma_{\hat{x}}^2 = \mathbb{V}\left[\frac{1}{\hat{N}_0}\sum_{k = 1}^{\hat{N}_0} f(\hat{x},\zeta^{k}_{h_0})\right]}$.\\
    \textit{Step 5}: Estimate $\displaystyle{N = \left\lceil\left(1+\frac{1}{2\alpha}\right) \frac{\sigma_{\hat{x}}^2}{\epsilon^2}\right\rceil}$.\\
    \textit{Step 6}: If $N>\hat{N}_0$, generate $N-\hat{N}_0$ extra samples of the random variable $\zeta_{h_0}$.\\
    \textit{Step 7}: Solve the optimization problem,\\
    \[
    \hat{\mathfrak{p}}^{*,h_0} = \min_{x\in\mathcal{X}}\frac{1}{N}\sum_{k = 1}^{N} f(x,\zeta^{k}_{h_0}).
    \]
 \textbf{Return}:  $\displaystyle{\hat{\mathfrak{p}}^{*,h_0}}.$
\end{algorithm}
\begin{algorithm}
    \caption{MLMC SAA}
    \label{algo:mlmc_saa}
    \textbf{Input}: $\epsilon,h_0,m,\hat{N}_0,\alpha,\beta$\\ 
    \textit{Step 1}: Generate $\zeta^1_{h_0},\zeta^2_{h_0},\dots,\zeta^{\hat{N}_0}_{h_0}$ independent and identically distributed sample of the random variable $\zeta_{h_0}$.\\
    \textit{Step 2}: Estimate $\displaystyle{\hat{x} = \arg\min_{x\in \mathcal{X}}\frac{1}{\hat{N}_0}\sum_{k = 1}^{\hat{N}_0} f(x,\zeta^{k}_{h_0})}$.\\
    \textit{Step 3}: Estimate $\displaystyle{h, L, \{N_{\ell}\}_{0\leq \ell \leq L}}$ using the formulas in Table \ref{table:mlmc_paramter}.\\
    \textit{Step 4}: For $\ell = 0,1,\dots,L$, generate $(\zeta^1_{\ell},\zeta^1_{\ell-1}),(\zeta^2_{\ell},\zeta^2_{\ell-1}),\dots,(\zeta^{N_{\ell}}_{\ell},\zeta^{N_{\ell}}_{\ell-1})$ independent and identically distributed samples of $(\zeta_{\ell},\zeta_{\ell-1})$.\\
    \textit{Step 5}: Solve the optimization problem,\\
    \[
    \hat{\mathfrak{p}}^{*,L} = \min_{x\in\mathcal{X}}\sum_{\ell = 0}^L\frac{1}{N_{\ell}}\sum_{k = 1}^{N_{\ell}} \left(f(x,\zeta^{k}_{\ell}) - f(x,\zeta^{k}_{\ell-1})\right).
    \]
 \textbf{Return}:  $\displaystyle{\hat{\mathfrak{p}}^{*,L}}.$
\end{algorithm}
\begin{table}[H]
\centering
\begin{tabular}{|c|>{\centering\arraybackslash}m{0.8\textwidth}|}
\hline
$m$ & $m_{\ell} = m^{\ell}, \; \ell = 0, \ldots, L $ \\
\hline
$L = L^*(\epsilon)$ & 
$\displaystyle{1 + \left\lceil \displaystyle{\frac{\log \left( (1+2\alpha)^{\frac{1}{2\alpha}} \left( \frac{\left| c_1 \right|} {\epsilon} \right)^{\frac{1}{\alpha}} h_0 \right)}{\log(m)}} \right\rceil}$ \\
\hline
$h = h^*(\epsilon)$ & 
$\displaystyle{\frac{h_0} {\left\lceil h_0 \left( 1 + 2\alpha \right)^{\frac{1}{2\alpha}} \left( \frac{\left| c_1 \right|} {\epsilon} \right)^{\frac{1}{\alpha}} m^{-L} \right\rceil }}$ \\
\hline
$q = q^*(\epsilon)$ & 
$\displaystyle{q_0(\epsilon) = \frac{1}{q_{\epsilon}^{\dagger}}, \quad q_{\ell}(\epsilon) = \frac{\lambda h^{\frac{\beta}{2}} \left( m_{\ell-1}^{-1} - m_{\ell}^{-1} \right)^{\frac{\beta}{2}}}{q_{\epsilon}^{\dagger} \sqrt{m_{\ell-1}+ m_{\ell}}}, \quad \ell = 1, \ldots, L,}$ \\
& \text{with } $\displaystyle{q_{\epsilon}^{\dagger} \text{ s.t. } \sum_{\ell=0}^{L} q_{\ell}(\epsilon) = 1 \text{ and } \lambda = \frac{V_1(\hat{x})}{V_h(\hat{x})}.}$ \\
\hline

$N = N^*(\epsilon)$ & 
$\displaystyle{\left \lceil \frac{\left( 1 + \frac{1}{2\alpha} \right) V_h(\hat{x}) q_{\epsilon}^{\dagger} \left( 1 + \lambda h^{\frac{\beta}{2}} \sum_{\ell=1}^{L} \left( m_{\ell-1}^{-1} - m_{\ell}^{-1} \right)^{\frac{\beta}{2}} \sqrt{m_{\ell-1}+ m_{\ell}} \right)}{\epsilon^2 }\right \rceil}$ \\
\hline
$N_{\ell} = N_{\ell}^*(\epsilon)$ & $\lceil N^*(\epsilon)q_{\ell}(\epsilon)\rceil$\\
\hline
\end{tabular}
\caption{Parameters for MLMC-SAA}
\label{table:mlmc_paramter}
\end{table}
The reader may refer to the Practitioner's Corner of section 9.5.2 in \cite{pages2018numerical} for a discussion on the calibration of the parameters $V_1(\hat{x})$,$V_h(\hat{x})$ and $c_1$. Remember that we do all the calculations after estimating $\hat{x}$ in step 2 in algorithm \ref{algo:mlmc_saa}.

\section{Appendix}
\label{appendix:B}
\begin{proof}[Proof (Proposition \ref{proposition:optimal_gap_mc})]
    Observe that as a consequence of the triangle inequality we have,
\[
\norm{\mathfrak{G}^h_N(\hat{x}) - \mathfrak{G}(\hat{x})}_2 = \norm{\mathfrak{G}^h_N(\hat{x}) -\mathfrak{G}^h(\hat{x})+\mathfrak{G}^h(\hat{x})- \mathfrak{G}(\hat{x})}_2 \leq \norm{\mathfrak{G}^h_N(\hat{x}) -\mathfrak{G}^h(\hat{x})}_2 + \norm{\mathfrak{G}^h(\hat{x})- \mathfrak{G}(\hat{x})}_2.
\]
The second term in the above inequality is bounded as,
\[
\norm{\mathfrak{G}^h(\hat{x})- \mathfrak{G}(\hat{x})}_2 \leq \norm{F_h(\hat{x}) - F(\hat{x})}_2 + \norm{\mathfrak{p}^{*,h} - \mathfrak{p}^{*}}_2
\]
Now $\displaystyle{\norm{F_h(\hat{x}) - F(\hat{x})}_2 \leq \sup_{x \in \mathcal{X}} \abs{\mathbb{E}\left(f(x, \zeta_h) - f(x, \zeta)\right)}\leq c_1h^{\alpha}}$, and 
\[\displaystyle{\norm{\mathfrak{p}^{*,h} - \mathfrak{p}^{*}}_2\leq \left(\expectation{\left[\sup_{x\in\mathcal{X}}|\expectation{[f(x,\zeta_h)]}-\expectation{[f(x,\zeta)]}|^2\right]}\right)^{1/2}}\leq c_1h^{\alpha},\]
therefore, $\displaystyle{\norm{\mathfrak{G}^h(\hat{x})- \mathfrak{G}(\hat{x})}_2\leq 2c_1h^{\alpha}}$. As for the first term \textit{i.e.},$\displaystyle{\norm{\mathfrak{G}^h_N(\hat{x}) -\mathfrak{G}^h(\hat{x})}_2}$ observe that,
\[
\norm{\mathfrak{G}^h_N(\hat{x}) -\mathfrak{G}^h(\hat{x})}_2 \leq \bnorm{\frac{1}{N}\sum_{k = 1}^N f(\hat{x},\zeta_h^k) - \expectation[f(\hat{x},\zeta_h)]}_2 + \bnorm{\min_{x \in \mathcal{X}}\frac{1}{N}\sum_{k = 1}^N f(x,\zeta_h^k) -\min_{x \in \mathcal{X}}\expectation[f(x,\zeta_h)]}_2,
\]
For the first term in the above equation, we have,
\begin{align}
    \bnorm{\frac{1}{N}\sum_{k = 1}^N f(x,\zeta_h^k) - \expectation[f(x,\zeta_h)]}_2 &= \left(\expectation{\left[\left(\frac{1}{N}\sum_{k = 1}^N f(\hat{x},\zeta_h^k) - \expectation[f(\hat{x},\zeta_h)\right)^2\right]}\right)^{\frac{1}{2}}\nonumber\\
    &=\left(\expectation{\left[\left(\frac{1}{N}\sum_{k = 1}^N \left(f(\hat{x},\zeta_h^k) - \expectation[f(\hat{x},\zeta_h)\right)\right)^2\right]}\right)^{\frac{1}{2}}\nonumber\\
    &=\left(\expectation{\left[\left(\frac{1}{N}\sum_{k = 1}^NZ_k^h(\hat{x})\right)^2\right]}\right)^{\frac{1}{2}}\nonumber\\
    &\leq \frac{\mathfrak{c}_2}{\sqrt{N}}\left(\expectation{[(f(\hat{x},\zeta_h)-\expectation[f(\hat{x},\zeta_h)])^2]}\right)^{1/2}\nonumber\\
    &\leq \frac{\mathfrak{c}_2}{\sqrt{N}}\sigma \nonumber
\end{align}
where the last two inequalities are the consequence of \cite{giles2019multilevel}(Lemma 2.5) and assumption \ref{variance_assumption}. Further, from the calculations in Theorem 5, we have,
\[
\bnorm{\min_{x \in \mathcal{X}}\frac{1}{N}\sum_{k = 1}^N f(x,\zeta_h^k) -\min_{x \in \mathcal{X}}\expectation[f(x,\zeta_h)]}_2 \leq  \mathfrak{c}_3\frac{\sigma}{\sqrt{N}}.
\]
By collating everything together and reassigning constants, we get the desired result.
\end{proof}
\begin{proof}[Proof (Proposition \ref{proposition:optimal_gap_mlmc})]
    To begin with, observe that from triangle inequality, we have,
\begin{align*}
\norm{\hat{\mathfrak{G}}^L(\hat{x}) - \mathfrak{G}(\hat{x})}_2 = \norm{\hat{\mathfrak{G}}^L(\hat{x}) - \mathfrak{G}^L(\hat{x})+\mathfrak{G}^L(\hat{x}) - \mathfrak{G}(\hat{x})}_2\leq \norm{\hat{\mathfrak{G}}^L(\hat{x}) - \mathfrak{G}^L(\hat{x})}_2 + \norm{\mathfrak{G}^L(\hat{x}) - \mathfrak{G}(\hat{x})}_2
\end{align*}
As before $\displaystyle{\norm{\mathfrak{G}^L(\hat{x}) - \mathfrak{G}(\hat{x})}_2 \leq 2c_1h_{L}^{\alpha}}$, therefore, we intend to study, $\displaystyle{\norm{\hat{\mathfrak{G}}^L(\hat{x}) - \mathfrak{G}^L(\hat{x})}_2}$. 
\begin{align*}
\bnorm{\hat{\mathfrak{G}}^L(\hat{x}) - \mathfrak{G}^L(\hat{x})}_2 &= \bnorm{\sum_{\ell = 0}^L \frac{1}{N_{\ell}}\sum_{k = 1}^{N_{\ell}}g(\hat{x},\bar{\zeta}^k_{\ell}) - \sum_{\ell = 0}^L \expectation{[g(\hat{x},\bar{\zeta}_{\ell})]} + \mathfrak{p}^{*,L} - \hat{\mathfrak{p}}^{*,L}}_2\\ &\leq \bnorm{\sum_{\ell = 0}^L \frac{1}{N_{\ell}}\sum_{k = 1}^{N_{\ell}}g(\hat{x},\bar{\zeta}^k_{\ell}) - \sum_{\ell = 0}^L \expectation{[g(\hat{x},\bar{\zeta}_{\ell})]}}_2 + \bnorm{\mathfrak{p}^{*,L} - \hat{\mathfrak{p}}^{*,L}}_2 
\end{align*}
Now,
\begin{align*}
\bnorm{\sum_{\ell = 0}^L \frac{1}{N_{\ell}}\sum_{k = 1}^{N_{\ell}}g(\hat{x},\bar{\zeta}^k_{\ell}) - \sum_{\ell = 0}^L \expectation{[g(\hat{x},\bar{\zeta}_{\ell})]}}_2 &\leq \bnorm{\frac{1}{N_0}\sum_{k = 1}^{N_0}\left(f(x,\zeta_0^k) - \expectation{f(x,\zeta_0)}\right)}_2\\
&+\bnorm{\sum_{\ell = 1}^L \frac{1}{N_{\ell}}\sum_{k = 1}^{N_{\ell}}\left(f(x,\zeta_{\ell}^k) - f(x,\zeta_{\ell-1}^k) - \expectation{[f(x,\zeta_{\ell}) - f(x,\zeta_{\ell-1})]}\right)}_2.
\end{align*}
Let 
\[
\mathcal{Z}_{\ell}(x) = \begin{cases}
    f(x,\zeta_0) - \expectation{f(x,\zeta_0)},\ell = 0\\
    f(x,\zeta_{\ell}) - f(x,\zeta_{\ell-1}) - \expectation{[f(x,\zeta_{\ell}) - f(x,\zeta_{\ell-1})]},\ell = 1,\dots,L
\end{cases}
\] then it is easy to $\mathcal{Z}_{\ell}(x)$ is a zero mean random variable, therefore by Lemma 2.5 in \cite{giles2019multilevel}, we have,
\begin{align*}
    \bnorm{\frac{1}{N_0}\sum_{k = 1}^{N_0}\left(f(x,\zeta_0^k) - \expectation{f(x,\zeta_0)}\right)}_2 &\leq \frac{\mathfrak{c}_2}{\sqrt{N_0}}\left(\expectation{[\mathcal{Z}_{0}(x)]^2}\right)^{1/2}\\
    \bnorm{\sum_{\ell = 1}^L \frac{1}{N_{\ell}}\sum_{k = 1}^{N_{\ell}}\left(f(x,\zeta_{\ell}^k) - f(x,\zeta_{\ell-1}^k) - \expectation{[f(x,\zeta_{\ell}) - f(x,\zeta_{\ell-1})]}\right)}_2 &\leq \mathfrak{c}_2\sum_{\ell = 0}^L\frac{1} {\sqrt{N_{\ell}}}\left(\expectation{[\mathcal{Z}_{\ell}(x)]^2}\right)^{1/2}.
\end{align*}
Consequently, by assumption \ref{assumption_mlmc_var}, we get,
\[
\bnorm{\sum_{\ell = 0}^L \frac{1}{N_{\ell}}\sum_{k = 1}^{N_{\ell}}g(\hat{x},\bar{\zeta}^k_{\ell}) - \sum_{\ell = 0}^L \expectation{[g(\hat{x},\bar{\zeta}_{\ell})]}}_2 \leq \sum_{\ell = 0}^L\frac{\mathfrak{c}_2}{\sqrt{N_{\ell}}}\left(\expectation{[\mathcal{Z}_{\ell}(x)]^2}\right)^{1/2} \leq \sum_{\ell = 0}^L\frac{c_2\mathfrak{c}_2}{\sqrt{N_{\ell}}}h_{\ell}^{\beta/2}.
\]
As for $\displaystyle{\bnorm{\mathfrak{p}^{*,L} - \hat{\mathfrak{p}}^{*,L}}_2 }$, we have from calculations in Theorem 6, that,
\[
\bnorm{\mathfrak{p}^{*,L} - \hat{\mathfrak{p}}^{*,L}}_2\leq 2c_2\bar{\mathfrak{c}}\sum_{\ell = 0}^{L}\frac{h_{\ell}^{\bar{\beta}/2}}{\sqrt{N_{\ell}}}
\]
where $\bar{\beta} = \beta\frac{1}{1+a}$. Clearly, for any $a>0$, $\bar{\beta}<\beta$, therefore we have,
\[
\bnorm{\hat{\mathfrak{G}}^L(\hat{x}) - \mathfrak{G}^L(\hat{x})}_2 \leq \mathfrak{c}_3\sum_{\ell = 0}^{L}\frac{h_{\ell}^{\bar{\beta}/2}}{\sqrt{N_{\ell}}}
\] for some constant $\mathfrak{c}_3$. Hence, the result follows.
\end{proof}
\end{appendices}
\end{document}